\documentclass[copyright,creativecommons]{eptcs}
\providecommand{\event}{ICE 2023} 

\usepackage{IEEEtrantools,cancel,textgreek,graphicx,eqparbox,mathtools,cases,pdfpages,amssymb,amsmath,amsthm,adjustbox,wrapfig,thmtools,subcaption}
\usepackage{iftex}

\ifpdf
  \usepackage{underscore}         
  \usepackage[T1]{fontenc}        
\else
  \usepackage{breakurl}           
\fi

\newcommand{\dq}{\ensuremath{\Delta{}\textrm{Q}}}

\newcommand{\dqsd}{\dq{}SD}

\newcommand{\SeqDelta}{\ensuremath{\mathbin{\bullet \hspace{-.6em} \rightarrow \hspace{-.85em} - \hspace{-.15em} \bullet}}}
\newcommand{\ProbChoice}[2]{\ensuremath{\mathbin{\substack{#1 \\ {\displaystyle \leftrightharpoons} \\ #2}}}}
\newcommand{\SingleProbChoice}[1]{\ensuremath{\mathbin{\substack{\left[#1\right] \\ {\displaystyle \leftrightharpoons}}}}}
\newcommand{\ProbChoiceSymb}{\ensuremath{\displaystyle \leftrightharpoons}}
\newcommand{\DenSem}[1]{\ensuremath{[\hspace{-0.15em}[#1]\hspace{-0.15em}]}}

\newcommand{\Base}{\ensuremath{\overline{\mathbb{B}}}}
\newcommand{\Outcomes}{\ensuremath{\mathbb{O}}}

\newcommand{\IRVs}{\ensuremath{\mathbb{I}}}
\newcommand{\NamedSem}[2]{\ensuremath{#1\DenSem{#2}}}
\newcommand{\DQAnalysis}[2]{\ensuremath{\NamedSem{\dq{}}{#1}}_{#2}}

\newcommand{\Operations}{\ensuremath{\mathbb{P}}}

\newcommand{\MultiToFinish}[1]{\ensuremath{\mathopen{\parallel^{#1}}}}
\newcommand{\MultiFTF}{\MultiToFinish{\exists}}
\newcommand{\MultiATF}{\MultiToFinish{\forall}}

\newcommand{\Observing}[1]{\ensuremath{\odot\hspace{-.4em}\odot #1}}
\newcommand{\ud}{\,\mathrm{d}}
\newcommand{\Dirac}{\textbf{\textdelta}}
\newcommand{\Properise}{\ensuremath{\eqmakebox[pft][c]{\rotatebox[origin=c]{90}{$\mathrlap{\rightarrow}\phantom{\rightarrow}\hspace{-.2em}|$}}}}

\definecolor{UnicornMilkCol}{rgb}{0.96,0.95,0.76}
\definecolor{SteelBlueCol}{rgb}{0.27,0.50,0.70}
\definecolor{SlateGrey}{rgb}{0.43,0.5,0.56}
\definecolor{DarkNavy}{rgb}{0.15, 0.25, 0.34}
\definecolor{purple}{rgb}{.5,.0,.5}

\declaretheorem[name=Remark,style=remark,qed={\lower-0.3ex\hbox{$\square$}}]{remark}
\theoremstyle{definition}
\newtheorem{definition}{Definition}
\newtheorem{notation}{Notation}
\newtheorem{example}{Example}
\AtBeginEnvironment{example}{%
  \pushQED{\qed}%
}
\AtEndEnvironment{example}{\popQED\endexample}
\theoremstyle{plain}
\newtheorem{theorem}{Theorem}
\newtheorem{proposition}{Proposition}

\newtheorem{lemma}{Lemma}

\title{Algebraic Reasoning About Timeliness}
\def\titlerunning{Algebraic Reasoning About Timeliness}
\def\authorrunning{Haeri et al.}

\author{Seyed Hossein HAERI
\institute{IOG, {Belgium}}
\institute{University of Bergen, {Norway}}
\email{hossein.haeri@iohk.io}
\and Peter W. THOMPSON
\institute{PNSol, {UK}}
\email{Peter.Thompson@pnsol.com}
\and Peter VAN ROY
\institute{Université catholique de Louvain, {Belgium}}
\email{pvr@info.ucl.ac.be}
\and Magne HAVERAAEN
\institute{University of Bergen, {Norway}}
\email{Magne.Haveraaen@uib.no}
\and Neil J. DAVIES
\institute{PNSol, {UK}}
\email{Neil.Davies@pnsol.com}
\and Mikhail BARASH
\institute{University of Bergen, {Norway}}
\email{mikhail.barash@uib.no}
\and Kevin HAMMOND
\institute{IOG, {UK}}
\email{kevin.hammond@iohk.io}
\and James CHAPMAN
\institute{IOG, {UK}}
\email{james.chapman@iohk.io}
}

\begin{document}
\maketitle
\begin{abstract}
  Designing distributed systems to have predictable performance under high load is difficult because of resource exhaustion, 
  non-linearity, and stochastic behaviour.
  Timeliness, i.e., delivering results within defined time bounds, is a central aspect of predictable performance.
  In this paper, we focus on timeliness using the \dq{} Systems Development paradigm (\dqsd{}, developed by PNSol), 
  which computes timeliness by modelling systems observationally using so-called outcome expressions.  
  An outcome expression is a compositional definition of a system’s observed behaviour in terms of its basic operations. 
  Given the behaviour of the basic operations, \dqsd{} efficiently computes the stochastic behaviour of the whole system including its timeliness.

  This paper formally proves useful algebraic properties of outcome expressions w.r.t. timeliness.
  We prove the different algebraic structures the set of outcome expressions form with the different \dqsd{} operators 
  and demonstrate why those operators do not form richer structures.
  We prove or disprove the set of all possible distributivity results on outcome expressions.
  On our way for disproving $8$ of those distributivity results, we develop a technique called \textit{properisation}, which gives rise to the first body of maths for \textbf{improper} random variables.
  Finally, we also prove $14$ equivalences that have been used in the past in the practice of \dqsd{}.

  An immediate benefit is rewrite rules that can be used for design exploration under established timeliness equivalence.
  This work is part of an ongoing project to disseminate and build tool support for \dqsd{}.
  The ability to rewrite outcome expressions is essential for efficient tool support.
\end{abstract}

\section{Introduction}

Designing distributed systems to have predictable performance under high load is difficult.
At high load, resources such as network, memory, storage, or CPU capacity will be exhausted, causing a dramatic effect on performance.
Prediction is difficult because the behaviour of system components and their interactions are both nonlinear and stochastic.
For over 20 years, a small group of people associated with the company PNSol has worked on diagnosing and designing systems to predict and correct performance problems \cite{pnsol}.
PNSol has developed the \dq{} Systems Development paradigm (\dqsd{}) as part of this work.
\dqsd{} has been used in areas as diverse as telecommunications \cite{TR-452.1} \cite{TR-452.2} \cite{Davi+Thom+Youn+Newt+Teig+Olde:2021}, WiFi \cite{Teig+Davi+Olav+Skei+Torr:2022}, and distributed ledgers \cite{Cout+Davi+Szam+Thom:2020}.
\dqsd{} has been applied to many large industrial systems, with clients including BT, Vodafone, Boeing Space and Defence, and IOG (formerly IOHK).

This paper defines and proves algebraic properties of the \dqsd{} operators w.r.t. timeliness, i.e., delivering outcomes within the acceptable time-frames.
In this paper, our sole resource of concern is time, although \dqsd{} includes other types of resources and their interaction.

This theoretical work is part of an ongoing project to disseminate and build tool support for \dqsd{}, to make it available to the wide community of system engineers.  
We base our work on the \dqsd{} formalisation given by Haeri et al. \cite{Haer+Thom+Davi+Roy+Hamm+Chap:2022}, which defines outcome expressions and their semantics, and gives a real-world example of \dqsd{} taken from the blockchain domain.

\subsection*{Contributions}

This paper gives a firm mathematical foundation for \dqsd, and uses this to establish important algebraic properties of the \dqsd{} operators with respect to timeliness, i.e., when the relevant resource is time.
This paper is based on a general model theory of resource analysis for systems specified using outcome expressions \cite{Haer+Thom+Roy+Have+Davi+Bara+Chap:2023}.
That model theory is the first of its kind and we specialise it using the timeliness analysis recipe that is commonly used in \dqsd{} (Definition~\ref{Defn:DQ.Analysis.Composition}).

\begin{itemize}
  \item
    We show that the set of outcome expressions forms different algebraic structures with the different \dqsd{} operators 
    (Theorems~\ref{Thrm:ProbChoice.Magma}--\ref{Thrm:FTF.Comm.Monoid}).
  \item
    We establish $3$ distributivity results in Section \ref{Sect:Distributivity} about the \dqsd{} operators (Theorem~\ref{Thrm:Dist.Prob}).
  \item
    We rule out the formation of certain richer algebraic structures by the set of outcome expressions and the current \dqsd{} operators 
    (Remarks~\ref{Rmrk:ATF.Not.Group}, \ref{Rmrk:FTF.Not.Group}, and \ref{Rmrk:FTF.ATF.No.Semiring}).
  \item
    We develop two new techniques for analysing the validity of algebraic equivalences: a new technique that we call \textit{Properisation} (Section~\ref{Sect:Properisation}) and another based on counterexamples (Section~\ref{Sect:Counterexamples}).
    We use those techniques to refute the remaining possible distributivity results in their full generality:
    $8$ using properisation (Theorem~\ref{Thrm:Dist.Ineq.Proper}) and $4$ using counterexamples (Theorem~\ref{Thrm:Dist.Fail.Counter}).

    From a mathematical viewpoint, properisation is an important contribution of ours.
    As far as we know, properisation is the first body of maths developed for \textit{improper random variable}s \cite{Triv:1982}.
  \item
    We provide guidelines for studying the necessary/sufficient conditions for the distributivity results we refute the generality of (Section~\ref{Sect:Potential.Distributivity}).
  \item
    We establish $14$ equivalences that have been used in the past in the practice of \dqsd{} (Section~\ref{Sect:Other.Equivalences}).
\end{itemize}
Full proofs can be found in the accompanying technical report  \cite{Haer+Thom+Roy+Have+Davi+Bara+Chap:2023}, which also shows how Fig.~\ref{Fig:Cache} can be further elaborated using code running in a Jupyter notebook.

The primary practical results of this paper are to establish distributive properties of \dqsd{} operators and other equivalences that are useful for rewriting outcome expressions. 
These enable common sub-expressions to be moved, for example, to reduce representational complexity, with associated gains in tool performance. 
Rewriting can also be used to produce normal forms, and, in particular, to extract reliability/failure probabilities without fully evaluating the outcome expression.
More generally, it can be used to establish equivalences between different designs with respect to their timeliness, even though their usage of other resources might differ, thereby allowing design exploration under equivalence.

\section{Motivating Example: Cache Memory} \label{Sect:Cache}

We give an example of a memory system consisting of a local cache with a remote main memory.
This example serves two purposes:
First, it shows how outcome diagrams can be used to model nontrivial systems.
Second, it shows the usefulness of the algebraic transformations of this paper.
We give the block diagram and the outcome diagram for this example.
We show how to rewrite (a simplified version of) the outcome diagram to swiftly compute failure (and, hence, success) rate of this design, giving it a `back of an envelope' feasibility test.
As we revisit this example later on, we will see that all this is possible because of the algebraic results proved in this paper.

Fig. \ref{Fig:Cache-system} gives the block diagram of the memory system.
A read message enters the cache; a cache hit -- when the memory word is in the cache -- results in an immediate return message; a cache miss -- when the memory word is not in the cache -- results in a main memory read.  The main memory is across a network, so accessing it requires communication in both directions.  Main memory access is guarded by a timeout in case of communication failure.
The cache miss initialises the timeout timer; the \textit{mreturn} message is passed through if it occurs before the timeout;
otherwise, a \textit{timeout} message is passed instead. 
Furthermore, there is a small probability that the remote main memory read fails.

\begin{figure}[t]
  \centering
  \hrule
  \vspace{1ex}
  \includegraphics[width=0.75\linewidth,trim=0 0 0 0,clip]{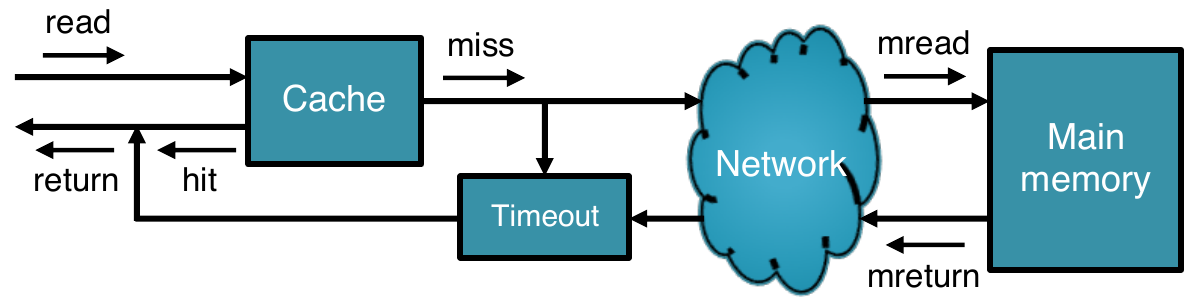}
  \vspace{1ex}
  \hrule
  \caption{Block Diagram for a Cache with Networked Main Memory}
  \label{Fig:Cache-system}
\end{figure}

\begin{figure}[t]
  \centering
  \hrule
    \adjustbox{trim={.0\width} {.1\height} {.0\width} {.33\height},clip}{\includegraphics[width=0.95\linewidth]{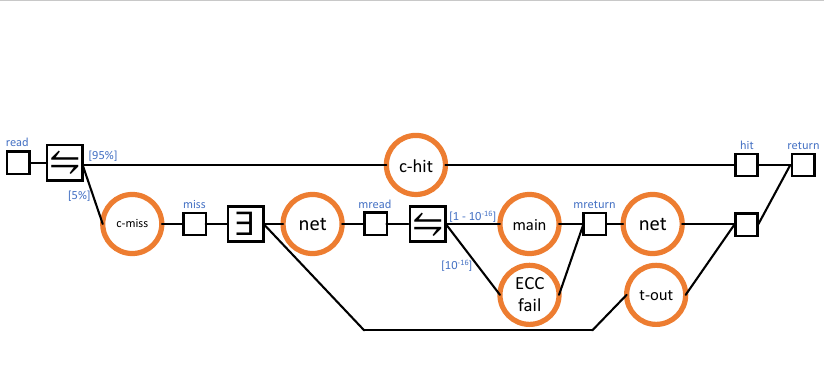}}
  \hrule
  \caption{Outcome Diagram for the Cache of Figure \ref{Fig:Cache-system}}
  \label{Fig:Cache}
\end{figure}

\paragraph{Outcome Diagram for the Cache with Networked Memory}

Fig. \ref{Fig:Cache} shows the outcome diagram for the memory system.
We can define an \textit{outcome} as what the system obtains by performing one of its tasks.
Outcomes are shown using orange circles in the outcome diagrams.
When there is a left-to-right path from one outcome to another, the right one is causally dependent on the left one.
Small square boxes show the starting and terminating sets of events of the corresponding outcomes.
Large square boxes are operators.
In Fig.~\ref{Fig:Cache} there are two \textit{probabilistic choices}, ``\ProbChoiceSymb'', and one \textit{first-to-finish} synchronisation, ``$\exists$''.
We assume that the cache hit rate is $95\%$.
That is modelled using the leftmost probabilistic choice with two paths, one to each outcome (``cache hit'' and ``cache miss''), decorated with their corresponding probabilities.
Timeout is modelled by a first-to-finish relationship between the main memory read and the timer.
We assume that the main memory uses Error-Correction Codes (ECC) to catch bit errors, but nevertheless
account for the possibility that a main memory access fails (e.g. because of hardware failure)
by giving it a failure rate of $10^{-16}$.
This assumption is modelled in Fig.~\ref{Fig:Cache} as a probabilistic choice between the ``main'' and ``ECC fail'' outcomes.

\begin{wrapfigure}{l}{0.5\textwidth}
  \vspace{-2em}
  \begin{center}
  \hrule
    \includegraphics[width=0.49\textwidth]{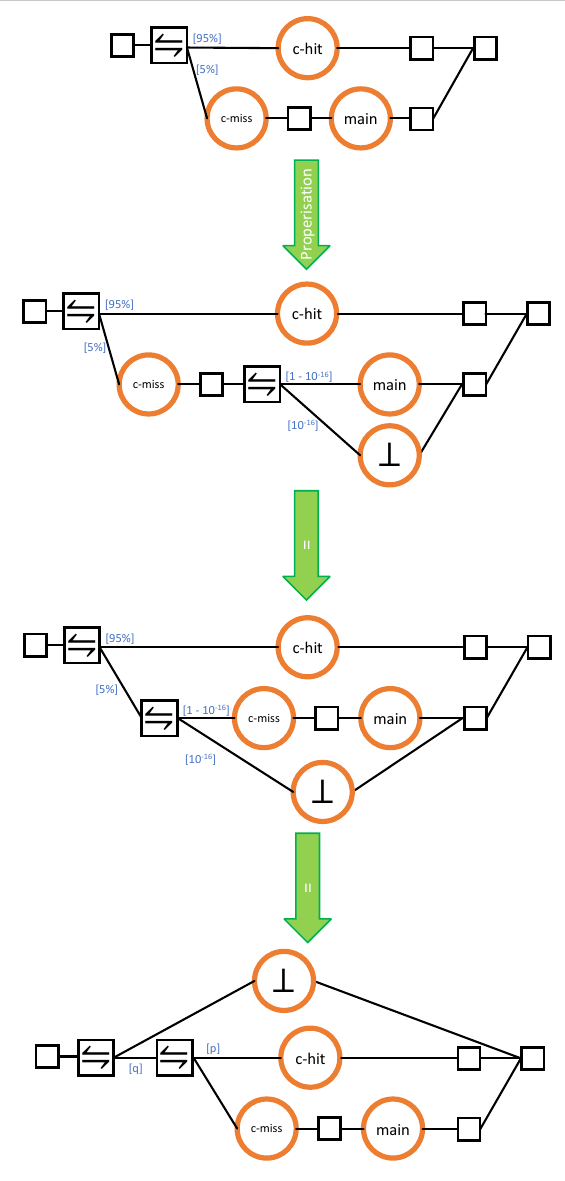}
  \hrule
  \end{center}
  \caption{Steps for Swiftly Calculating the Failure Rate}
  \label{Fig:Cache.Steps}
  \vspace{-3em}
\end{wrapfigure}

\paragraph{Failure Rate}

Let us now compute the failure rate by doing algebraic transformations as defined in this paper.
Without loss of generality, we can assume that the network has zero delay and the timeout is infinite.
One can then simplify Figure~\ref{Fig:Cache} to the outcome diagram at the top of Figure~\ref{Fig:Cache.Steps}.
In that diagram, the ECC failure is hidden in the failure rate assigned to $\mathit{main}$ in the timeliness analysis of the diagram.
However, as we will prove in this paper, one can also \textit{properise} $\mathit{main}$ and explicitly demonstrate that failure rate as a probabilistic choice, \textbf{whilst retaining the level of timeliness}.
The result of that properisation is shown in the second diagram from the top, where ``$\bot$'' represents \textit{(unconditional) failure}.

According to the developments of this paper, one can rewrite the second diagram from the top to the third and then to the bottom one, again, \textbf{whilst retaining the level of timeliness}.
What is important about the bottom diagram of Fig.~\ref{Fig:Cache.Steps} is that it comprises of a probabilistic choice between failure and everything else in the diagram.
As will be proved later, for some $p$, and for $q = (1-0.05 \times 10^{-16}) = 0.999999999999999995$, we have swiftly obtained the failure rate.
Those numbers immediately tell the system engineer that, under the current assumptions about cache hit and main memory failure rates, \textbf{every} implementation will be infeasible if the overall success rate must be greater than $q$.

The techniques used for this example generalise in a straightforward fashion to any system modelled using an outcome diagram.

In the remainder of this paper, Examples~\ref{Xmpl:Syntax}--\ref{Xmpl:Properisation} will come back to the developments of this section by supplying syntax, semantics, timeliness analysis, and authorising the rewrite steps taken here.

\paragraph{Closing Remarks on the Example}

Realistic cache memories are often more complex than this example, which gives rise to more complicated outcome diagrams in which ``$\bot$'' will appear at multiple depths.
Thanks to the results we prove in this paper, techniques such as that of this section can be used to accumulate those $\bot$s. 

While the probabilities in this example may seem small, they can combine with probabilities from other parts of the system, and it is important to be able to keep track of them.
Dismissing them as `minimal' risks missing potentially serious failures when many `small' probabilities aggregate.

\section{Background} \label{Sect:Background}

\begin{figure}[h]
  \centering
  \hrule
  \vspace{1ex}
  \includegraphics[width=0.8\linewidth,trim=0 0 0 0.25in,clip]{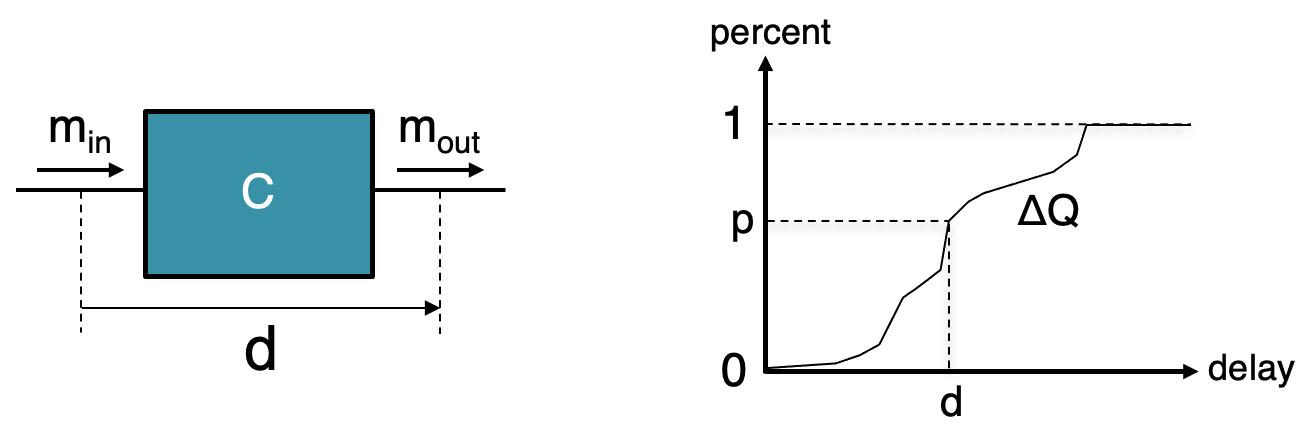}
  \hrule
  \caption{A Component's Operation and its Cumulative Delay Function}
  \label{Fig:Outcome-Delay}
\end{figure}

\begin{wrapfigure}{r}{.4\textwidth}
\begin{subfigure}{.4\textwidth}
  \vspace{-1.5em}
  \begin{center}
  \hrule
    \includegraphics[width=\linewidth]{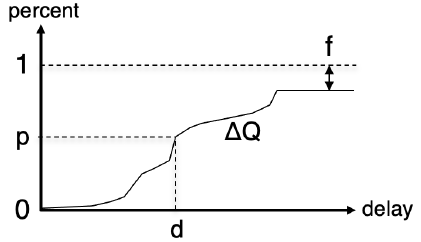}
  \hrule
  \end{center}
  \caption{Failure is modelled as a quality attenuation whose limit is less than $1$.}
  \label{Fig:Outcome-Failure}
\end{subfigure}
\begin{subfigure}{.4\textwidth}
  \begin{center}
  \hrule
    \includegraphics[width=\linewidth]{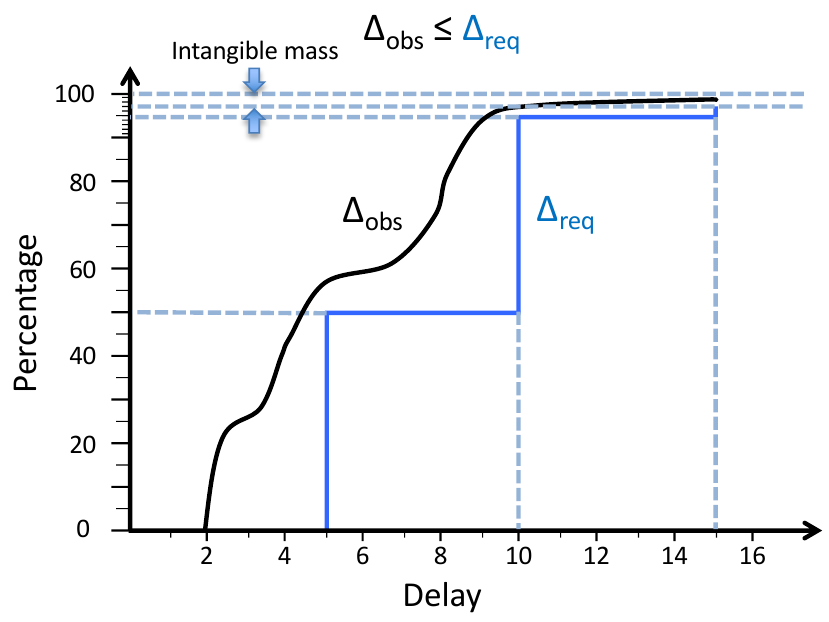}
  \hrule
  \end{center}
  \caption{Timeliness: $\dq_{obs}$ (the observed quality attenuation \dq{}) is always to the left and above $\dq_{req}$ (the required \dq{}).}
  \label{Fig:Timeliness}
\end{subfigure}
\hrule
\caption{Failure and Timeliness}
\vspace{-4em}
\end{wrapfigure}

\paragraph{Outcome and Quality Attenuation}

Consider a component C which inputs message $m_{\mbox{\em in}}$ and outputs message $m_{\mbox{\em out}}$ after a delay $d$.
Doing this many times will usually give different delays.
We define a cumulative delay function so that $p$ percent of delays 
are less or equal to $d$.
Figure \ref{Fig:Outcome-Delay} gives an illustration.


The \dqsd{} paradigm generalises this simple measurement.
We measure delay not only for messages, but for all system behaviours that have a starting event and a terminating event.
Given a starting event $e_{\mbox{\em in}}$ and a terminating event $e_{\mbox{\em out}}$, what the system gains within the $(e_{\mbox{\em in}},e_{\mbox{\em out}})$ time frame is called an instance of an {\em outcome}.
We also generalise the property that we measure: we measure not only delay, but any property that makes the system less than perfect.
The cumulative distribution function of the property is then called a {\em quality attenuation} and is denoted by a \dq.
In what follows, we will consistently use the terms outcome and quality attenuation.

\paragraph{Failure}

It is straightforward to generalise the quality attenuation to model both delay and failure.
It suffices to allow the cumulative delay function's limit to be less than $1$.
Figure \ref{Fig:Outcome-Failure} illustrates this possibility.
There is an $f$ percent probability that the delay is infinite, which corresponds precisely to a failure.
For the component, it means simply that there is an input message $m_{\mbox{\em in}}$ with no corresponding output message $m_{\mbox{\em out}}$.
Mathematically, the delay is modelled by a random variable that is allowed to be \textbf{improper}:
The probability that it is infinite can be greater than $0$.
This probability is called the intangible mass of the Improper Random Variable (IRV) \cite{Triv:1982}.

The ability to model delay and failure as a single quantity is a key strength of \dqsd{}.
It makes it easy to explore trade-offs between delay and failure in the system design.
This ability shows up clearly in the algebra presented in this paper.

\paragraph{Timeliness} \label{Sect:Timeliness.Paragraph}

We define {\em timeliness} as a relation (defined in \cite{TR-452.1}) between an observed $\dq_{\mbox{\em obs}}$ and a required $\dq_{\mbox{\em req}}$.
We say that the system {\em satisfies timeliness} for a given outcome if $\dq_{\mbox{\em obs}} \leq \dq_{\mbox{\em req}}$.
Figure \ref{Fig:Timeliness} illustrates this condition.

\paragraph{Outcome Expressions}

For a system consisting of multiple interconnected components, one can define a graph that combines all the components' outcomes.
This graph defines the causal relationships between the outcomes and is called an {\em outcome diagram}.
Each outcome diagram has a corresponding \textit{outcome expression} -- a mathematical description of the diagram
\footnote{
In this paper, we take the equivalence between the outcome expressions and outcome diagrams for granted.
That equivalence is not the focus of this paper.}.
Given an outcome expression and the quality attenuations of all its components, it is possible to compute the quality attenuation of the complete system.
The reverse process can also be fruitful:
given an outcome expression and the required quality attenuation of the complete system, one can estimate the required quality attenuations of its components.
This gives the system designer a powerful tool for both design and diagnosis.

Outcome expressions can be manipulated according to algebraic rules, in particular those presented in this paper,
which are useful to system designers using \dqsd{}.
As part of an ongoing project, we are building software tools to support \dqsd{}, which can use 
the algebraic rules presented here for symbolic manipulation of outcome expressions.

\subsection*{\dqsd{}}

\dqsd{} is a systems development paradigm that is able to compute many system properties
early on in the design process, such as performance (latency and throughput), timeliness, resource consumption, risks, and feasibility.
\dqsd{} is used both for diagnosis and design:
\begin{itemize}
  \item {\em System Diagnosis}.
    \dqsd{} can analyse an existing system, to pinpoint anomalous behaviours so their origin can be found and the system can be corrected.
  \item {\em System Design}.
    \dqsd{} can estimate performance trade-offs during the design process.
    At \textbf{every} step of the design process, performance of the complete system can be estimated by a computation on the partial design.
    This computation also determines whether or not the system is feasible, i.e., whether it can or cannot meet the requirements.
\end{itemize}
While historically \dqsd{} has primarily been used to diagnose and correct problems in large industrial systems, PNSol has recently used \dqsd{} to design the Shelley block diffusion algorithm as used in the Cardano blockchain \cite{Haer+Thom+Davi+Roy+Hamm+Chap:2022}.
More information on \dqsd{} can be found in a tutorial given at HiPEAC 2023 \cite{HiPEAC:2023}.

\section{An Algebraic Perspective on Timeliness} \label{Sect:Perspective}

\subsection{Syntax of Outcome Expressions}

\begin{definition}[Haeri et al.~\cite{Haer+Thom+Davi+Roy+Hamm+Chap:2022}] \label{Defn:Outcome.Syntax}
  Assume a set $\Base$ of primitive outcomes.
  We use variables $\beta \in \Base$ to represent individual primitive outcomes.
  We define the abstract syntax of outcome expressions as follows:
  $$
  \begin{array}{r@{\ \ }c@{\ \ }llc@{\ \ }ll}
    \Outcomes \ni o  & ::= & \beta                    & \text{primitive outcome}\\
                     & |   & o \SeqDelta o'           & \text{sequential composition} & |   & (o \parallel^\forall o') & \text{all-to-finish (a.k.a. last-to-finish)}\\
                     & |   & o \ProbChoice{m}{m'} o'  & \text{probabilistic choice} & |   & (o \parallel^\exists o') & \text{any-to-finish (a.k.a. first-to-finish).}
  \end{array}
  $$
\end{definition}
This defines outcome expressions as combinations of primitive outcomes $\beta$ and four composition operators.
In the case of probabilistic choice, $m$ and $m'$ are numeric weights which give the probabilities of choosing the left or right alternative, respectively.
For convenience, we also introduce another notation $o \ProbChoice{[p]}{} o'$ where the probability $(1-p)$ for the right alternative is implied.
We distinguish two constant outcomes: $\top$ for ``perfection'' and $\bot$ for ``unconditional failure.''

Note that the operator ``$\exists$" in the outcome diagrams is ``$\MultiFTF$" in the outcome expressions.
That is to signify that when two outcomes are connected by first-to-finish, they are performed concurrently; hence the ``$\parallel$" sign.
One need not emphasise that concurrency in the outcome diagrams because our left-to-right directional convention on causal dependency already implies concurrency when forking off an ``$\exists$'' in the outcome diagrams.
Similarly, for ``$\MultiATF$'' in the outcome expressions, the sign in the outcome diagrams is simply ``$\forall$''.

\begin{example} \label{Xmpl:Syntax}
  Getting back to our motivating example, we can now transcribe the outcome diagram of Fig.~\ref{Fig:Cache} into an outcome expression:
  \begin{equation} \label{Eqtn:Cache.1}
    \textit{c-hit} \ProbChoice{[95\%]}{} 
      (\textit{c-miss} \SeqDelta
        ((\mathit{net}  \SeqDelta (\mathit{main} \ProbChoice{[1 - 10^{-16}]}{} \bot) \SeqDelta \mathit{net}) \parallel^\exists \textit{t-out}))\text{.}
  \end{equation}
  We will use this outcome expression in further examples.
\end{example}

\subsection{Timeliness Semantics for Outcome Expressions} \label{Sect:Time.Semantics}

Let $\dq{}(x)$ denote the probability that an outcome occurs in a time $t \leq x$.
In order to represent both delay and failure in a single quantity,
a \dq{} is represented by an improper random variable (IRV), allowing the total probability not to reach 100\% \cite{Triv:1982}.
The \textit{intangible mass} of such an IRV is $\Im(\dq{}) = 1 - \lim_{x \to \infty} \Delta Q(x)$.
For a given \dq{}, the intangible mass $\Im(\dq{})$ encodes the probability of exceptions or failure occurring.

Denote the set $\IRVs$ of all IRVs that are differentiable and the values of which are always greater than or equal to zero.
Statistically speaking, every $\iota \in \IRVs$ can be represented both using its Probability Density Function (PDF) or its Cumulative Distribution Function (CDF),
where the former is the derivative of the latter.
For convenience, we will freely switch between the two representations as the need rises.
Fix a countable set of \dq{} variables $\Delta_v$.
We define $\Delta = \Delta_v \cup \IRVs$ to denote both IRVs and \dq{} variables.
When $\delta \in \Delta$ is in its CDF representation, we write $\delta'$ for its derivative, which is the PDF representation.

We first define a mapping between primitive outcomes $\Base$ and \dq{}s.
\begin{definition}
  We call a function $\NamedSem{\Delta_\circ}{.}: \Base \rightarrow \Delta$ a {\em basic assignment} when $\NamedSem{\Delta_\circ}{\top} = \mathbf{1}$ and $\NamedSem{\Delta_\circ}{\bot} = \mathbf{0}$, where $\mathbf{1}$ and $\mathbf{0}$ are the functions always returning the constants $1$ and $0$, respectively.
\end{definition}

\begin{example} \label{Exmp:Cache.Basic.Assignment}
  Every timeliness analysis of Fig.~\ref{Fig:Cache} \textit{\`a la} \dqsd{} needs a basic assignment that at least has mappings for the five individual outcomes in Equation~(\ref{Eqtn:Cache.1}) (namely, $\mathit{c\text{-}hit}$, $\mathit{c\text{-}miss}$, $\mathit{t\text{-}out}$, $\mathit{net}$, and $\mathit{main}$) so that \dq{}$_\mathit{c\text{-}hit}$, \dq{}$_\mathit{c\text{-}miss}$, \dq{}$_\mathit{t\text{-}out}$, \dq{}$_\mathit{net}$, and \dq{}$_\mathit{main}$, are known initially.
  That is generally possible because:
  the first two are properties of the cache;the timeout is chosen by the designer; 
  the network performance is known; and the main memory read time (and failure rate) is also known.
\end{example}

We now define the semantics of an outcome expression as a mapping between the outcome expression and an IRV, for a given basic assignment.
\begin{definition}[Haeri et al.~\cite{Haer+Thom+Davi+Roy+Hamm+Chap:2022}] \label{Defn:DQ.Analysis.Composition}
  Given a basic assignment $\NamedSem{\Delta_\circ}{.}: \Base \rightarrow \Delta$, define $\DQAnalysis{.}{\Delta_\circ}: \mathbb{O} \rightarrow \IRVs$ such that
  $$
  \begin{array}{l@{\ =\ }l}
    \DQAnalysis{\beta}{\Delta_\circ} &  \left\{ \begin{array}{ll}
      \mathbf{1}                     & \text{when } \NamedSem{\Delta_\circ}{\beta} \notin \IRVs\\ 
      \NamedSem{\Delta_\circ}{\beta} & \text{otherwise}
    \end{array} \right.\\
    
    \DQAnalysis{o \SeqDelta o'}{\Delta_\circ} & \DQAnalysis{o}{\Delta_\circ} \ast \DQAnalysis{o'}{\Delta_\circ}\\
    
    \DQAnalysis{o \ProbChoice{m}{m'} o'}{\Delta_\circ} & \frac{m}{m + m'}\DQAnalysis{o}{\Delta_\circ} + \frac{m'}{m + m'}\DQAnalysis{o'}{\Delta_\circ}\\
    
    \DQAnalysis{o\ \MultiATF\ o'}{\Delta_\circ} & \DQAnalysis{o}{\Delta_\circ} \times \DQAnalysis{o'}{\Delta_\circ}\\
    
    \DQAnalysis{o\ \MultiFTF\ o'}{\Delta_\circ} & \DQAnalysis{o}{\Delta_\circ} + \DQAnalysis{o'}{\Delta_\circ} - \DQAnalysis{o}{\Delta_\circ} \times \DQAnalysis{o'}{\Delta_\circ} \\
  \end{array}
  $$
\end{definition}
Here, the notation $\ast$ denotes the convolution of two $\dq{}$s.
In the above formulae, the random variables are always represented using their CDFs except for sequential composition, where the representation is PDFs on both sides. Note that the PDF of $\top$ is the Dirac $\Dirac$ function.
In what follows, we will drop $\Delta_\circ$ whenever the basic assignment is fixed throughout a computation. 

One way to interpret Definition~\ref{Defn:DQ.Analysis.Composition} is that $\DQAnalysis{.}{\Delta_\circ}$ is a homomorphism from the term algebra of outcome expressions $\mathbb{O}$ to an algebra of probability distributions $\IRVs$.

\begin{remark}
  Note that, according to Definition~\ref{Defn:DQ.Analysis.Composition}, we get $\DQAnalysis{o_1 \SeqDelta o_2}{} = \DQAnalysis{o_2 \SeqDelta o_1}{}$.
  This may seem counter-intuitive because  $o_1 \SeqDelta o_2 \neq o_2 \SeqDelta o_1$. 
  $\DQAnalysis{o_1 \SeqDelta o_2}{} = \DQAnalysis{o_2 \SeqDelta o_1}{}$ is, nonetheless, valid because, intuitively, $o_1 \SeqDelta o_2$ is just as timely as $o_2 \SeqDelta o_1$.
  See the proof of Theorem~\ref{Thrm:SeqDelta.Monoid}~\cite{Haer+Thom+Roy+Have+Davi+Bara+Chap:2023} for the mathematical justification of that intuition.
\end{remark}

\subsection{Motivating Example: Timeliness Analysis}

\begin{example} \label{Xmpl:Cache.DQ.Work}
  Given the developments of Example~\ref{Exmp:Cache.Basic.Assignment}, here is how to work out the quality attenuation of Fig.~\ref{Fig:Cache} using Definition~\ref{Defn:DQ.Analysis.Composition}:
  Take $\mathit{mem} = \mathit{net}  \SeqDelta (\mathit{main} \ProbChoice{[1 - 10^{-16}]}{} \bot) \SeqDelta \mathit{net}$ to be the outcome of the networked main memory read.
  We start by computing \dq{}$_\mathit{mem}$:
  \begin{equation}
    \dq{}_\mathit{mem} = 
      \dq{}_\mathit{net} \ast 
      ((1-10^{-16}) \times \dq{}_\mathit{main} + 10^{-16} \times \dq{}_\bot) \ast
      \dq{}_\mathit{net}
  \end{equation}
  which, because \dq{}$_\bot = \mathbf{0}$, we can simplify to:
  \begin{equation}
    \dq{}_\mathit{mem} = 
      \dq{}_\mathit{net} \ast 
      (1-10^{-16}) \times \dq{}_\mathit{main} \ast
      \dq{}_\mathit{net}
  \end{equation}
  The overall $\dq{}_\mathit{obs}$ is then given by:
  \begin{equation} \label{Eqtn:DQ.Obs}
    \dq{}_\mathit{obs} = 0.95 \times \dq{}_\textit{c-hit} +  
      0.05 \times (\dq{}_\textit{c-miss} \ast (\dq{}_\mathit{mem} + \dq{}_\textit{t-out} - \dq{}_\mathit{mem} \times \dq{}_\textit{t-out}))\text{.}
  \end{equation}
  This computation gives us the CDF for the execution time of a memory read.
  The numeric computation is easily performed by a software tool.
  For readers interested in seeing fully worked-out numerical examples, we recommend looking up the tutorial \cite{HiPEAC:2023}.
\end{example}

Recall that, in Section~\ref{Sect:Timeliness.Paragraph}, we defined timeliness as $\dq{}_\mathit{obs} \le \dq{}_\mathit{req}$ (this relation is a partial order, defined in \cite{Haer+Thom+Davi+Roy+Hamm+Chap:2022}).
Definition~\ref{Defn:DQ.Analysis.Composition} gives this more context.
Using Definition~\ref{Defn:DQ.Analysis.Composition}, the systems engineer can work out the $\dq{}_\mathit{obs}$ of an outcome so they can compare the result against the required $\dq{}_\mathit{req}$.

\begin{example} \label{Xmpl:Cache.Better.Than}
  Given the developments of Example~\ref{Xmpl:Cache.DQ.Work}, we can now get back to the plots in Fig.~\ref{Fig:Timeliness}.
  Taking the blue plot for $\dq_{\mathit{req}}$, the cache outcome diagrams developed in Section~\ref{Sect:Cache} are timely so long as $50\%$ of the queries submitted to the cache can be handled within $5$ units of time, $95\%$ of them in $10$ units, and $97\%$ in $15$ units.
  Furthermore, the cache is fine to drop $3\%$ of the queries.\footnote{
    This is an instance of the ``better-than'' part of the Quantitative Timeliness Agreement (QTA) \cite{TR-452.1}.}

  Taking the black plot in Fig.~\ref{Fig:Timeliness} as that of Equation~(\ref{Eqtn:DQ.Obs}) after insertion of real numbers, our designed cache is timely enough because it always handles the queries within the acceptable time frame and drops less queries than the acceptable maximum.
  In other words, it has less delay and less failure rate.
  Visually, that amounts to the black plot always being to the left and above the blue plot.
\end{example}

\subsection{Connecting Algebra to Timeliness} \label{Sect:Model.Theory}

In our accompanying technical report \cite{Haer+Thom+Roy+Have+Davi+Bara+Chap:2023}, we give a model theoretic formulation for studying the algebraic properties of resource consumption.
This paper focuses on time as its sole resource of interest and uses that formulation for time exclusively without getting into the technical details of the formulation itself.

An algebraic structure often consists of a carrier set, a few operations on the carrier set, and a finite set of identities that those operations need to satisfy.
Given our focus on timeliness \`a la \dqsd{}, the carrier set will always be \Outcomes\ in this paper.
The full set of operators on \Outcomes\ is $\{\SeqDelta, \MultiATF, \MultiFTF, \ProbChoice{}{}\}$.
However, most algebraic structures do not need all those operators.
Different structures work with different number of operations;
for example, a monoid works with only one operation; whilst a group works with two.
Finally, the identities are of the form $o_l = o_r$.

We take $\DQAnalysis{.}{}$ (Definition~\ref{Defn:DQ.Analysis.Composition}) to be the model of time consumption for \Outcomes.
We write
\begin{itemize}
  \item $\Observing{time} \vDash o_l = o_r$ when $\DQAnalysis{o_l}{} = \DQAnalysis{o_r}{}$.
    That is when $o_l$ and $o_r$ are as timely.
  \item $\Observing{time} \vDash (\Outcomes, P): s$ for an algebraic structure $s$ and a set of \dqsd{} operators $P$ when \linebreak $\Observing{time} \vDash o_l = o_r$, for every equation $o_l = o_r$
    \begin{itemize}
      \item that is constructed using the operators in $P$, and
      \item that is required for the formation of $s$.
    \end{itemize}
\end{itemize}

With time being our solo resource of interest in this paper, we will drop the initial ``$\Observing{time} \vDash$'' from the above formulation hereafter.

\section{Algebraic Structures}

This section establishes several important properties on \Outcomes{}:
\begin{itemize}
  \item probabilistic choice forms a magma (Theorem~\ref{Thrm:ProbChoice.Magma});
  \item sequential composition forms a commutative monoid with $\top$ and $\bot$ as the identity and absorbing elements (Theorem~\ref{Thrm:SeqDelta.Monoid});
  \item all-to-finish forms a commutative monoid with $\top$ and $\bot$ as the identity and absorbing elements (Theorem~\ref{Thrm:ATF.Comm.Monoid});
  \item any-to-finish forms a commutative monoid with $\bot$ and $\top$ as the identity and absorbing elements (Theorem~\ref{Thrm:FTF.Comm.Monoid}); and
  \item neither all-to-finish nor any-to-finish nor their combination form the familiar richer algebraic structures (Remarks~\ref{Rmrk:ATF.Not.Group}, \ref{Rmrk:FTF.Not.Group}, and \ref{Rmrk:FTF.ATF.No.Semiring}).
\end{itemize}

\begin{theorem} \label{Thrm:ProbChoice.Magma}
  $(\Outcomes, \ProbChoice{}{})$ forms a magma when observing time.
\end{theorem}

\noindent
A magma is the weakest algebraic structure.
That is because $\ProbChoice{}{}$ is not even associative.
Despite this, expressions containing two consecutive occurrences of $\ProbChoice{}{}$
can still be re-associated. However, in this case the coefficients will change.
Lemmas~\ref{Lemm:Prob.Choice.Assoc.Left.Coef} and \ref{Lemm:Prob.Choice.Assoc.Right.Coef} give the exact formulae.

\begin{theorem} \label{Thrm:SeqDelta.Monoid}
  $\Observing{\text{time}} \vDash (\Outcomes, \SeqDelta): forms\ a\ \text{commutative monoid}$ with $\top$  and $\bot$ as the identity and absorbing elements, respectively.
\end{theorem}

\begin{theorem} \label{Thrm:ATF.Comm.Monoid}
  $\Observing{\text{time}} \vDash (\Outcomes, \parallel^\forall): forms\ a\ \text{commutative monoid}$ with $\top$ and $\bot$ as the identity and absorbing elements, respectively.
\end{theorem}

\begin{remark} \label{Rmrk:ATF.Not.Group}
  It is important to notice that, when observing time, $(\Outcomes, \parallel^\forall)$ does \emph{not} form a group.
  That is because, in general, an outcome has no inverse element - intuitively, one can never undo an outcome!
  
  In order to prove that claim formally, suppose otherwise.
  That is, suppose that there exist a pair of outcomes $o_1$ and $o_2$ such that $o_1\ \MultiATF\ o_2 = \top$.
  Then, $\DQAnalysis{o_1\ \MultiATF\ o_2}{} = \DQAnalysis{\top}{}$ which implies $\delta_1 \times \delta_2 = \mathbf{1} \Rightarrow \delta_2 = \frac{\mathbf{1}}{\delta_1}$.
  However, given that $\delta_1 \le \mathbf{1}$, we get $\delta_2 \ge \mathbf{1}$.
  The latter inequality can only be satisfied when $o_1 = \top$.
  Restricting the application of \dqsd{} to perfection is not practical. 
\end{remark}

\begin{theorem} \label{Thrm:FTF.Comm.Monoid}
  $\Observing{\text{time}} \vDash (\Outcomes, \parallel^\exists): forms\ a\ \text{commutative monoid}$ with $\bot$ and $\top$ as the identity and absorbing elements, respectively.
\end{theorem}

\begin{remark} \label{Rmrk:FTF.Not.Group}
  Similar to the case for $\parallel^\forall$, it is important to note that, when observing time, $(\Outcomes, \parallel^\exists)$ does not form a group.
  Again, it is the lack of an inverse element that is causing the trouble.
  Following our previous result,   
  suppose that there exist a pair of outcomes $o_1$ and $o_2$ such that $o_1\ \MultiFTF\ o_2 = \bot$.
  Then, $\DQAnalysis{o_1\ \MultiFTF\ o_2}{} = \DQAnalysis{\bot}{}$ which implies $\delta_1 + \delta_2 - \delta_1 \times \delta_2 = \mathbf{0} \Rightarrow \delta_2 = \frac{\delta_1}{\delta_1 - 1}$.
  However, because $\delta_1 \le \mathbf{1}$, we get $\delta_2 \le \mathbf{0}$.
  But, only $\bot$ can satisfy the latter inequality.
  There is no reason to develop a system in which all the outcomes will fail unconditionally!
\end{remark}

Having established that both $(\Outcomes, \parallel^\forall)$ and $(\Outcomes, \parallel^\exists)$ form commutative monoids for time, a natural question is whether $(\Outcomes, \parallel^\forall, \parallel^\exists)$ or $(\Outcomes, \parallel^\exists, \parallel^\forall)$ form semi-rings.
This is not the case, since they do not distribute over one another.

Lemma~\ref{Lemm:FTF.Top} helps Remark~\ref{Rmrk:FTF.ATF.No.Semiring} show how the desirable distributivities fail.

\begin{lemma} \label{Lemm:FTF.Top}
  $\Observing{\text{time}} \vDash o_1\ \MultiFTF\ o_2 = \top$ implies $o_1 = \top$ and $o_2 = \top$.
\end{lemma}

\begin{remark} \label{Rmrk:FTF.ATF.No.Semiring}
  Neither $(\Outcomes, \parallel^\forall, \parallel^\exists)$ nor $(\Outcomes, \parallel^\exists, \parallel^\forall)$ form a semi-ring when observing time:
  for this to be the case, $\parallel^\forall$ and $\parallel^\exists$ would need to distribute over one another.
%
  The first distributivity requirement is:
  \begin{equation} \label{Eqn:Dist.FTF.Over.ATF}
    o_1 \parallel^\exists (o_2 \parallel^\forall o_3) \stackrel{?}{=} (o_1 \parallel^\exists o_2) \parallel^\forall (o_1 \parallel^\exists o_3)
  \end{equation}
  Equating $\DQAnalysis{.}{}$s of the two sides, one eventually makes it to the requirement that either $\delta_1 = \mathbf{0}$ or $\DQAnalysis{(o_1 \parallel^\exists o_3) \parallel^\exists o_2}{} = \top$.
  In other words, it follows by Lemma~\ref{Lemm:FTF.Top} that Equation~(\ref{Eqn:Dist.FTF.Over.ATF}) can only hold under the trivial conditions when either $o_1 = \bot$ or $o_1 = o_2 = o_3 = \top$.
  The second distributivity requirement is
  \begin{equation} \label{Eqn:Dist.ATF.Over.FTF}
    o_1 \parallel^\forall (o_2 \parallel^\exists o_3) \stackrel{?}{=} (o_1 \parallel^\forall o_2) \parallel^\exists (o_1 \parallel^\forall o_3)
  \end{equation}
  Again, equating $\DQAnalysis{.}{}$s of the two sides, one eventually comes to observe that Equation~(\ref{Eqn:Dist.ATF.Over.FTF}) only holds 
  when $\delta_1 = \mathbf{1} \wedge \delta_2 \neq \mathbf{0} \wedge \delta_3 \neq \mathbf{0}$, i.e., when $o_1 = \top \wedge o_2 \neq \bot \wedge o_3 \neq \bot$.
\end{remark}

\section{Equivalences Containing Constant Outcomes} \label{Sect:Other.Equivalences}

\begin{figure*}
  \centering
  \hrule
    $$
    \begin{array}{c@{\hspace{1.9em}}c@{\hspace{1.8em}}c@{\hspace{1.9em}}c}
      \bot \ProbChoice{}{} \bot = \bot & 
      (o_1 \ProbChoice{}{} \bot) \SeqDelta o_2 = (o_1 \SeqDelta o_2) \ProbChoice{}{} \bot &
      o \SeqDelta \bot = \bot & \top \ProbChoice{}{} \top = \top\\
      
      \bot \SeqDelta o = \bot &
      o_1 \SeqDelta (o_2 \ProbChoice{}{} \bot) = (o_1 \SeqDelta o_2) \ProbChoice{}{} \bot  &
      \top \SeqDelta o = o & o \SeqDelta \top = o\\
      
      \top \MultiATF o = o &
      (o_1 \ProbChoice{}{} \top) \SeqDelta o_2 = (o_1 \SeqDelta o_2) \ProbChoice{}{} o_2 &
      \multicolumn{2}{c}{o_1 \SeqDelta (o_2 \ProbChoice{}{} \top) = (o_1 \SeqDelta o_2) \ProbChoice{}{} o_1}\\
      
      \bot \MultiFTF o = o &
      o_1 \ProbChoice{[p]}{} (o_2 \ProbChoice{[q]}{} \top) = o_2 \ProbChoice{[q(1 - p)]}{} (o_1 \SingleProbChoice{\frac{p}{1 - q(1 - p)}} \top)  &
      \multicolumn{2}{c}{\bot \ProbChoice{[p]}{} (\bot \ProbChoice{[q]}{} o) = \bot \ProbChoice{[p + (1 - p) q]}{} o}
    \end{array}
    $$
  \hrule
  \caption{Equivalences Containing $\top$ and $\bot$}
  \label{Fig:Equivalences}
\end{figure*}

\dqsd{} is already in use by its practitioners, who, amongst other usages, simplify outcome expressions according to their timeliness analysis.
In particular, Figure~\ref{Fig:Equivalences} distils a list of equivalences that are used in such simplifications.
Those equivalences all contain constant outcomes ($\top$ or $\bot$).

Equivalences of Figure~\ref{Fig:Equivalences} provide the basis for rewrite rules that are useful for construction of normal forms, such as expressing a given system as a convolution of probabilistic choices or a probabilistic choice of convolutions.
Such rewriting allows for: extraction of common sub-expressions permitting aggregation of failure rates (distinguishing between conditional and non-conditional failure); identifying minimal delays; and highlighting branching probabilities to identify issues of relative criticality.
This is useful for quickly assessing whether a particular outcome decomposition is \emph{feasible} without having to compute the complete \dq{}.
See Section~\ref{Sect:Cache}, for example.
In addition, the equivalences of Figure~\ref{Fig:Equivalences} are very handy in the proofs of properties such as those established in this paper.
Two examples, amongst many, are the proofs of Theorem~\ref{Thrm:Dist.Fail.Counter} and Lemma~\ref{Lemm:Fail.Acc}.

Before we delve into Figure~\ref{Fig:Equivalences}, we prove a result about re-associating probabilistic choice.
Given an expression with two consecutive probabilistic choices, one of which wrapped inside a pair of parentheses, the \dqsd{} practitioner might be interested in wrapping the other two inside a pair of parentheses -- re-associating the probabilistic choices, in effect.
Lemmata~\ref{Lemm:Prob.Choice.Assoc.Left.Coef} and \ref{Lemm:Prob.Choice.Assoc.Right.Coef} give the conditions on the coefficients of those probabilistic choices.

\begin{lemma} \label{Lemm:Prob.Choice.Assoc.Left.Coef}
  $o_1 \ProbChoice{[p]}{} (o_2 \ProbChoice{[q]}{} o_3) = (o_1 \ProbChoice{[p']}{} o_2) \ProbChoice{[q']}{} o_3$ iff $p' = \frac{p}{1 - (1- p)(1 - q)}$ and $q' = 1 - (1- p)(1 - q)$.
\end{lemma}

\begin{lemma} \label{Lemm:Prob.Choice.Assoc.Right.Coef}
  $(o_1 \ProbChoice{[p]}{} o_2) \ProbChoice{[q]}{} o_3 = o_1 \ProbChoice{[p']}{} (o_2 \ProbChoice{[q']}{} o_3)$ iff $p' = pq$ and $q' = \frac{q(1 - p)}{1 - pq}$.
\end{lemma}

\begin{theorem} \label{Thrm:Fig.Equivalences}
  The equivalences in Fig.~\ref{Fig:Equivalences} are correct.
\end{theorem}
\begin{proof}
  We will only present the proof of $\bot \ProbChoice{m_1}{m_2} \bot = \bot$ here.
  The rest of the equivalences are proved similarly:
  $$
    \DQAnalysis{\bot \ProbChoice{m_1}{m_2} \bot}{} = \frac{m_1}{m_1 + m_2} \mathbf{0} + \frac{m_2}{m_1 + m_2} \mathbf{0} = \mathbf{0} = \DQAnalysis{\bot}{}\text{.}
  $$
\end{proof}

\begin{remark}
  The very last equivalence in Fig.~\ref{Fig:Equivalences} was incorrectly formulated (though never published) prior to this paper.
  Thanks to the formalisation developed in \cite{Haer+Thom+Davi+Roy+Hamm+Chap:2022}, that mistake was corrected, and the equivalences have been given a sound footing.
\end{remark}

\subsection{Motivating Example: Correctness of the Three Bottom Rewrites}

\begin{example} \label{Xmpl:Cache.Last.3.Steps}
  We are now in position to confirm the steps taken in Fig.~\ref{Fig:Cache.Steps}.
  Note first that, after dismissing the back-and-forth network connections and the timeout, Equation~(\ref{Eqtn:Cache.1}) simplifies to
  \begin{equation} \label{Eqtn:Cache.2}
    \textit{c-hit} \ProbChoice{[95\%]}{} 
      (\textit{c-miss} \SeqDelta (\mathit{main} \ProbChoice{[1 - 10^{-16}]}{} \bot))
  \end{equation}
  which, according to Theorem~\ref{Thrm:Fig.Equivalences}, is equivalent to
  \begin{equation} \label{Eqtn:Cache.3}
     \textit{c-hit} \ProbChoice{[95\%]}{} ((\textit{c-miss} \SeqDelta \mathit{main})\ProbChoice{[1 - 10^{-16}]}{} \bot)
  \end{equation}
  which, again, can be rewritten using Lemma~\ref{Lemm:Prob.Choice.Assoc.Left.Coef} as
  \begin{equation} \label{Eqtn:Cache.4}
    (\textit{c-hit} \ProbChoice{[.]}{} (\textit{c-miss} \SeqDelta \mathit{main})) \ProbChoice{[q]}{} \bot
  \end{equation}
  for $q = (1-0.05 \times 10^{-16}) = 0.999999999999999995$.
  Equations~(\ref{Eqtn:Cache.2}), (\ref{Eqtn:Cache.3}), and (\ref{Eqtn:Cache.4}) are the outcome expressions for the bottom three outcome diagrams of Fig.~\ref{Fig:Cache.Steps}, respectively.
\end{example}

\section{Distributivity} \label{Sect:Distributivity}

In this section, we consider the distributivity results between the \dqsd{} operators.
Recall that out of the four $\Operations$ operators, three are commutative (i.e., $\SeqDelta$, $\MultiATF$, and $\MultiFTF$) and one is not (i.e., $\ProbChoiceSymb$).
Hence, it is only possible for right- and left-distributivity to differ when $\ProbChoiceSymb$ is the outermost operator.
That gives rise to $2 \times \binom{3}{1} + \binom{3}{1}\binom{3}{1} = 15$ possible ways for distributing $\Operations$ operators over each other.
Theorem~\ref{Thrm:Dist.Prob} establishes $3$ of those $15$.
In Section~\ref{Sect:Potential.Distributivity}, we show how the routine technique for examining the equivalence of expressions (i.e., equating the $\DQAnalysis{.}{}$ of the two sides) is not that helpful for the study of the remaining $12$ distributivity results.
That leads to Sections~\ref{Sect:Counterexamples} and \ref{Sect:Properisation}, which disprove the generality of $4$ and $8$ distributivity results using counterexamples (Theorem~\ref{Thrm:Dist.Fail.Counter}) and properisation (Theorem~\ref{Thrm:Dist.Ineq.Proper}), respectively.

We use the following syntactic convention:
when, in an equivalence, two $\ProbChoice{}{}$s are used without weights, each on precisely one side of the equivalence, we will assume that the weights of those $\ProbChoice{}{}$s are the same.
We therefore do not bother to repeat those weights.
For example, in the theorem below, there exist weights $m_2$ and $m_3$ such that $o_2 \ProbChoice{m_2}{m_3} o_3$ and $(o_1 \SeqDelta o_2) \ProbChoice{m_2}{m_3} (o_1 \SeqDelta o_3)$, but we omit these.

\begin{theorem} \label{Thrm:Dist.Prob}
  Let $o_1, o_2, o_3 \in \Outcomes$ and $p \in \{\SeqDelta, \MultiATF, \MultiFTF\}$.
  Then,
  \begin{itemize}
    \item $\Observing{\text{time}} \vDash o_1\ p\ (o_2 \ProbChoice{}{} o_3) = (o_1\ p\ o_2) \ProbChoice{}{} (o_1\ p\ o_3)$, and
    \item $\Observing{\text{time}} \vDash (o_1 \ProbChoice{}{} o_2)\ p\ o_3 = (o_1\ p\ o_3) \ProbChoice{}{} (o_2\ p\ o_3)$.
  \end{itemize}
  
\end{theorem}

\subsection{Potential Distributivity} \label{Sect:Potential.Distributivity}

As we are going to see in Sections~\ref{Sect:Counterexamples} and \ref{Sect:Properisation}, the remaining $12$ potential distributivity results do not hold \textbf{in general}.
Nevertheless, this section uses the routine technique for studying the equivalence of expressions:
Equating the $\DQAnalysis{.}{}$ of the two sides.
That is important because:
\begin{itemize}
  \item firstly, it shows why the routine technique does not help, thereby motivating the next sections;
  \item secondly, it presents some of the necessary conditions for those distributivity results to hold.
    Although pretty immature, such conditions help the \dqsd{} practitioner to verify, under special circumstances, whether their given IRVs can satisfy the provided conditions.
\end{itemize}
We do not know of better necessary conditions for the remaining $12$ results (\emph{if indeed they are soluble at all}).
In this section, we demonstrate the necessary conditions of one distributivity result out the $12$.

We begin by Proposition~\ref{Prop:Seq.PDF.to.CDF}, which is a simple yet handy result.

\begin{proposition} \label{Prop:Seq.PDF.to.CDF}
  Suppose that $o_1 = o_2 \SeqDelta o_3$.
  Then, $\Observing{\text{time}} \vDash \delta_1(t) = \int (\delta'_2 \ast \delta'_3)(t) \ud t$.
\end{proposition}

When observing time, for
\begin{equation} \label{Eqn:Right.Dist.Prob.Seq}
  (o_1 \SeqDelta o_2) \ProbChoice{m}{m'} o_3 \stackrel{?}{=} (o_1 \ProbChoice{m}{m'} o_3) \SeqDelta (o_2 \ProbChoice{m}{m'} o_3)
\end{equation}
to hold, according to Proposition~\ref{Prop:Seq.PDF.to.CDF},
\begin{IEEEeqnarray}{rCl}
  \DQAnalysis{(o_1 \SeqDelta o_2) \ProbChoice{m}{m'} o_3}{} &
    = &
    \frac{m}{m + m'} \int (\delta_1' \ast \delta_2')(t) \ud t + \frac{m'}{m + m'} \delta_3\nonumber\\
  & = &
    \frac{m}{m + m'} \iint \delta_1'(\tau) \delta_2'(t - \tau) \ud \tau \ud t + \frac{m'}{m + m'} \delta_3 \qquad \label{Eqn:Right.Dist.Prob.Seq.LHS}
\end{IEEEeqnarray}
and
\begin{IEEEeqnarray}{rCl}
  \IEEEeqnarraymulticol{3}{l}{\DQAnalysis{(o_1 \ProbChoice{m}{m'} o_3) \SeqDelta (o_2 \ProbChoice{m}{m'} o_3)}{} = \int \left(\frac{m}{m + m'} \delta'_1 + \frac{m'}{m + m'} \delta'_3\right)\ \ast \left(\frac{m}{m + m'} \delta'_2 + \frac{m'}{m + m'} \delta'_3\right)(t) \ud t \hspace{3.5em}} \nonumber\\
  \IEEEeqnarraymulticol{3}{r}{= \iint \left(\frac{m}{m + m'} \delta'_1(t) + \frac{m'}{m + m'} \delta'_3(t)\right)\ \times \left(\frac{m}{m + m'} \delta'_2(t - \tau) + \frac{m'}{m + m'} \delta'_3(t - \tau)\right) \ud \tau \ud t \text{.}\hspace{2em}}\label{Eqn:Right.Dist.Prob.Seq.RHS}
\end{IEEEeqnarray}
For Equation~(\ref{Eqn:Right.Dist.Prob.Seq}) to hold, the right-hand-sides of Equations~(\ref{Eqn:Right.Dist.Prob.Seq.LHS}) and (\ref{Eqn:Right.Dist.Prob.Seq.RHS}) need to be equal.
That is,
\begin{IEEEeqnarray}{rCl}
  \IEEEeqnarraymulticol{2}{l}{\frac{m}{m + m'} \iint \delta_1'(\tau) \delta_2'(t - \tau) \ud \tau \ud t + \frac{m'}{m + m'} \delta_3} & = \hspace{4em} \nonumber\\
  \IEEEeqnarraymulticol{3}{l}{\qquad\iint \left(\frac{m}{m + m'} \delta'_1(t) + \frac{m'}{m + m'} \delta'_3(t)\right)\ \times \left(\frac{m}{m + m'} \delta'_2(t - \tau) + \frac{m'}{m + m'} \delta'_3(t - \tau)\right) \ud \tau \ud t \quad} \label{Eqn:Right.Dist.Prob.Seq.Con}
\end{IEEEeqnarray}
This is a differential equation for which we do not know a general solution.
Given particular values for $\delta_1$, $\delta_2$, and $\delta_3$, however, the \dqsd{} practitioner might be able to solve it. 

\subsection{Counterexamples} \label{Sect:Counterexamples}

As will be worked out in Remark~\ref{Rmrk:FTF.Fail.Acc}, properisation does not quite work for outcome expressions containing $\MultiFTF$ because $\bot$ is not compositional under $\MultiFTF$.
In this section, we present a less advanced yet effective technique for refuting distributivity results: counterexamples.
A single counterexample suffices to refute an equivalence.
That is how Theorem~\ref{Thrm:Dist.Fail.Counter} refutes $4$ distributivity results out of the questionable $12$ (in their full generality).

\begin{theorem} \label{Thrm:Dist.Fail.Counter}
  For every $o_1, o_2, o_3 \in \Outcomes$,
  $$
  \begin{array}{l@{\hspace{4em}}l}
    o_1 \ProbChoice{}{} (o_2 \MultiFTF o_3) \neq (o_1 \ProbChoice{}{} o_2) \MultiFTF (o_1 \ProbChoice{}{} o_3) & (o_1 \MultiFTF o_2) \ProbChoice{}{} o_3 \neq (o_1 \ProbChoice{}{} o_3) \MultiFTF (o_2 \ProbChoice{}{} o_3)\\
    o_1 \MultiFTF (o_2 \SeqDelta o_3) \neq (o_1 \MultiFTF o_2) \SeqDelta (o_1 \MultiFTF o_3) & o_1 \SeqDelta (o_2 \MultiFTF o_3) \neq (o_1 \SeqDelta o_2) \MultiFTF (o_1 \SeqDelta o_3) \text{.}
  \end{array}
  $$
\end{theorem}
\begin{proof}
  We only prove the last item here.
  The other inequalities can be proved similarly using the same technique.
  Take $o_2 = o_3 = \top$ and let $\DQAnalysis{o_1}{} = \delta_1$.
  By Theorem~\ref{Thrm:Fig.Equivalences}, $o_1 \SeqDelta (o_2 \MultiFTF o_3) = o_1 \SeqDelta (\top \MultiFTF \top) = o_1 \SeqDelta \top = o_1$.
  Therefore,
  \begin{equation} \label{Eqtn:Dist.Seq.F2F.Fail.1}
    \DQAnalysis{o_1 \SeqDelta (o_2 \MultiFTF o_3)}{} = \delta_1\textnormal{.}
  \end{equation}
  On the other hand, by Theorem~\ref{Thrm:Fig.Equivalences}, $(o_1 \SeqDelta o_2) \MultiFTF (o_1 \SeqDelta o_3) = (o_1 \SeqDelta \top) \MultiFTF (o_1 \SeqDelta \top) = o_1 \MultiFTF o_1$.
  Thus,
  \begin{equation} \label{Eqtn:Dist.Seq.F2F.Fail.2}
    \DQAnalysis{(o_1 \SeqDelta o_2) \MultiFTF (o_1 \SeqDelta o_3)}{} = \delta_1 + \delta_1 - \delta_1\delta_1\textnormal{.}
  \end{equation}
  Equations~(\ref{Eqtn:Dist.Seq.F2F.Fail.1}) and (\ref{Eqtn:Dist.Seq.F2F.Fail.2}) together imply
  $
    \delta_1 = 2\delta_1 - \delta_1^2 \Rightarrow \delta_1 = \mathbf{0} \vee \delta_1 = \mathbf{1} \Rightarrow o_1 = \bot \vee o_1 = \top
  $.
  The result follows because, for any other $o_1$ and $o_2 = o_3 = \top$, the two sides will not be equal.
\end{proof}

\section{Properisation} \label{Sect:Properisation}

This section sets the stage using Theorem~\ref{Thrm:Properise} for a technique that we call \textit{properisation} and use for disproving equivalences (in their full generality).

Properisation is based on the following important observation:
if two outcomes do not fail similarly, they are not equivalent.
Properisation is an algebraic technique for swiftly extracting the failure behaviour of outcomes via rewriting but without assessing the rest of their timeliness behaviour.
Once the failure parts of the timeliness behaviours are at hand for the two sides, one can check whether they are equal,
and if they are not, deduce that the outcomes in question are therefore unequal.

Our intuition for the choice of name ``properisation'' for this technique follows:
recall that $\dq{}$s are CDFs (or PDFs) of \textbf{im}proper random variables.
Properisation is a technique based on making the $\dq$ of an outcome $o$ proper (by scaling it) 
and restoring its amount of improperness -- i.e., $o$'s intangible mass, denoted by $\Im(\dq{}(o))$ -- 
as a probabilistic choice (of the right weights) between $o$ and $\bot$.
That is also the intention behind the symbol we use for properisation: ``$\Properise$.''
As one can see in Figure~\ref{Fig:Outcome-Failure}, the CDF of an improper random variable needs not to make it to the ``ceiling'' (i.e., $1$).
The symbol ``$\Properise$'' that we use is intended to resemble the act of `sticking the CDF to the ceiling' (represented by the horizontal bar at the top of ``$\Properise$'')!

Now, the formal definitions of properisation.

\begin{definition}
  For an $\iota \in \IRVs$ such that $\Im(\iota) = i$, write $\iota' = \iota\Properise$ when $\mathit{dom}(\iota) = \mathit{dom}(\iota')$  and $\iota'(x) = \frac{1}{1 - i}\iota(x)$ for every $x \in \mathit{dom}(\iota)$.
  Call $\iota'$ the properisation of $\iota$.
\end{definition}

\begin{proposition}
  $\Im(\iota\Properise) = 0$, for all $\iota \in \IRVs$.
\end{proposition}

Intuitively, for IRVs, ``$.\Properise$'' produces a scaled random variable with no intangible mass.

\begin{definition}
  Fix two basic assignments $\Delta, \Delta'$ and a base variable $\beta$ such that $\Delta(\beta) = \iota$.
  Write $\Delta' = \Delta\Properise^\beta$ when
  $$
  \begin{array}{l@{\hspace{1em}}l@{\hspace{10em}}l@{\hspace{1em}}l}
    \Delta'(\beta') = \Delta(\beta')  & \text{for } \beta' \neq \beta &
    \Delta'(\beta') = \iota\Properise & \text{otherwise.}
  \end{array}
  $$
  We say $\Delta\Properise^\beta$ is the result of \textit{properisation} of $\beta$ in $\Delta$.
\end{definition}

Intuitively, $\Delta\Properise^\beta$ produces a new basic assignment that is the as same $\Delta$ everywhere except $\beta$, where the assigned IRV is propoerised.

\begin{notation}
  Write $o[o'/\beta]$ for the familiar $\lambda$-Calculus notation for substitutions: $o$ in which every instance of $\beta$ is replaced by $o'$.
\end{notation} 

\begin{definition} \label{Defn:Properise}
  Fix a basic assignment $\Delta$ and a base variable $\beta$ such that $\Delta(\beta) = \iota$ where $\Im(\iota) = i$.
  Write $(o, \Delta)\Properise^\beta = (o', \Delta')$ when $o' = o[(\beta \ProbChoice{[1 - i]}{} \bot) / \beta]$ and $\Delta' = \Delta\Properise^\beta$.
  We say that $o'$ is the result \textit{properisation} of $\beta$ in $o$ according to $\Delta$.
\end{definition} 

As a shorthand, we write $(o, \Delta)\Properise^{\beta_1, \beta_2}$ for $\left((o, \Delta)\Properise^{\beta_1}\right)\Properise^{\beta_2}$ and $\Delta\Properise^{\beta_1, \beta_2}$ for $\left(\Delta\Properise^{\beta_1}\right)\Properise^{\beta_2}$.

As one can see from Definition~\ref{Defn:Properise}, the act of properisation of a base variable $\beta$ in an outcome $o$ is according to a given basic assignment $\Delta$.
That is, the move from the right-hand-side of $(o, \Delta)\Properise^\beta = (o', \Delta')$ to its left-hand-side is performed by taking two steps in unison:
\begin{enumerate}
  \item scaling according to the intangible mass of $\Delta(\beta)$ so that $\beta$ is no longer improper in the resulting new basic assignment $\Delta'$; and,
  \item replacing every occurrence of $\beta$ in the outcome $o$ with the probabilistic choice that is weighted according to the intangible mass of $\Delta(\beta)$, resulting in the new outcome $o'$.
\end{enumerate}
The idea is that the intangible mass that $\Delta'$ takes away $o'$ returns, leaving timeliness intact.
Lemma~\ref{Lemm:Properise} utilises that idea.

\begin{lemma} \label{Lemm:Properise}
  Suppose that $(o, \Delta)\Properise^{\beta_1, \beta_2, \dots, \beta_n} = (o', \Delta')$ for some $\beta_1, \beta_2,$ $\dots, \beta_n \in \Base$, $o, o' \in \Outcomes$ and basic assignments $\Delta$ and $\Delta'$.
  Then, $\DQAnalysis{o}{\Delta} = \DQAnalysis{o'}{\Delta'}$.
\end{lemma}

Theorem~\ref{Thrm:Properise} utilises Lemma~\ref{Lemm:Properise} for examining equivalence of pairs of outcome expressions with no properisation relationship.

\begin{theorem} \label{Thrm:Properise}
  Suppose $\Delta$ and $\Delta'$ are two basic assignments.
  Suppose also that $o_1, o'_1, o_2, o'_2 \in \Outcomes$ such that $(o'_1, \Delta') = (o_1, \Delta)\Properise^{\beta_1, \beta_2, \dots, \beta_n}$ and $(o'_2, \Delta') = (o_2, \Delta)\Properise^{\beta_1, \beta_2, \dots, \beta_n}$, for some $\beta_1, \beta_2,$ $\dots, \beta_n \in \Base$.
  Then, $\DQAnalysis{o_1}{\Delta} = \DQAnalysis{o_2}{\Delta}$ iff $\DQAnalysis{o'_1}{\Delta'} = \DQAnalysis{o'_2}{\Delta'}\text{.}$
\end{theorem}

\subsection{Motivating Example: Correctness of the Properisation Step}

\begin{example} \label{Xmpl:Properisation}
  Recall from Section~\ref{Sect:Cache} that we took the failure rate of our ECC to be $10^{-16}$.
  One way to model that failure rate is to assume a basic assignment $\Delta$ such that $\Im(\Delta(\mathit{main})) = 10^{-16}$.
  Note also that the outcome expression for the top outcome diagram of Fig.~\ref{Fig:Cache.Steps} is
  \begin{equation*}
    \textit{c-hit} \ProbChoice{[95\%]}{} (\textit{c-miss} \SeqDelta \mathit{main})\text{.}
  \end{equation*}
  Furthermore, recall from Example~\ref{Xmpl:Cache.Last.3.Steps} that the outcome expression for the second diagram of Fig.~\ref{Fig:Cache.Steps} from the top is
  \begin{equation*}
    \textit{c-hit} \ProbChoice{[95\%]}{} 
      (\textit{c-miss} \SeqDelta (\mathit{main} \ProbChoice{[1 - 10^{-16}]}{} \bot))\text{.}
  \end{equation*}

  Now, suppose another basic assignment $\Delta' = \Delta\Properise^\mathit{main}$.
  Observe first that the latter outcome expression above is the properisation of $\mathit{main}$ in the former according to $\Delta$.
  Finally, thanks to Lemma~\ref{Lemm:Properise}, we know that one can rewrite the former outcome expression to the latter provided that one also replaces $\Delta$ with $\Delta'$.
  Hence, timeliness remains intact over taking the properisation step of Fig.~\ref{Fig:Cache.Steps}.
\end{example}

\subsection{Disproving the Remaining Distributivity Results}

Armed with Theorem~\ref{Thrm:Properise}, we can now outline the properisation technique:

Suppose two outcome expressions $o$ and $o'$ the equivalence of which is to be studied.
One begins by studying the equivalence of $o\Properise^{\beta_1, \dots, \beta_n}$ and $o'\Properise^{\beta_1, \dots, \beta_n}$ for some $\beta_1, \dots, \beta_n \in \Base$.
Now, suppose that -- after the application of algebraic laws -- one gets to rewrite $o\Properise^{\beta_1, \dots, \beta_n}$ to $(\dots) \ProbChoice{[p]}{} \bot$ and $o'\Properise^{\beta_1, \dots, \beta_n}$ to $(\dots) \ProbChoice{[p']}{} \bot$.
One concludes that $o \neq o'$ if one can show that $p \neq p'$.

We start the application of our properisation technique by obtaining some useful results.
Lemma~\ref{Lemm:Fail.Acc} paves the way for the applications of the above technique.
They instruct one on how to accumulate failure at the rightmost corner when the operator between two pairs of parentheses is $\SeqDelta$, $\ProbChoice{}{}$, and $\MultiATF$, respectively.
Unfortunately, $\MultiFTF$ has no such property, as will be shown by Remark~\ref{Rmrk:FTF.Fail.Acc}.

\begin{lemma} \label{Lemm:Fail.Acc}
  For every $o_1, o_2, o_3 \in \Outcomes$,
  $$
  \begin{array}{c}
    (o_1 \ProbChoice{\left[ p_1 \right]}{} \bot) \SeqDelta (o_2 \ProbChoice{\left[ p_2 \right]}{} \bot) = (o_1 \SeqDelta o_2) \ProbChoice{\left[ p_1p_2 \right]}{} \bot\\
    (o_1 \ProbChoice{[p_1]}{} \bot) \ProbChoice{[p]}{} (o_2 \ProbChoice{[p_2]}{} \bot) = (o_1 \ProbChoice{[q]}{} o_2) \ProbChoice{[r]}{} \bot \text{ where } q = \frac{pp_1}{p_2 - pp_2 + pp_1} \text{ and } r = p_2 - pp_2 + pp_1\\
    (o_1 \ProbChoice{[p_1]}{} \bot) \MultiATF (o_2 \ProbChoice{[p_2]}{} \bot) = (o_1 \MultiATF o_2) \ProbChoice{[p_1p_2]}{} \bot\text{.}
  \end{array}
  $$
\end{lemma}
\begin{proof}
  We only prove the first equivalence here.
  The proof is similar for the other two equivalences.
  
  By Theorems~\ref{Thrm:Dist.Prob} and \ref{Thrm:Fig.Equivalences},
  $$
    (o_1 \ProbChoice{\left[ p_1 \right]}{} \bot) \SeqDelta (o_2 \ProbChoice{\left[ p_2 \right]}{} \bot) =
    ((o_1 \ProbChoice{\left[ p_1 \right]}{} \bot) \SeqDelta o_2) \ProbChoice{\left[ p_2 \right]}{} \bot =
    ((o_1 \SeqDelta o_2) \ProbChoice{\left[ p_1 \right]}{} \bot) \ProbChoice{\left[ p_2 \right]}{} \bot =
    (o_1 \SeqDelta o_2) \ProbChoice{\left[ p_1p_2 \right]}{} \bot\text{.}
  $$
\end{proof}

\begin{remark} \label{Rmrk:FTF.Fail.Acc}
  Interestingly enough, there is no $p$ such that the following holds in its full generality:
  $$
    (o_1 \ProbChoice{[p_1]}{} \bot) \MultiFTF (o_2 \ProbChoice{[p_2]}{} \bot) \stackrel{?}{=} (o_1 \MultiFTF o_2) \ProbChoice{[p]}{} \bot\textnormal{.}
  $$
  Suppose there were such a $p$.
  One gets to observe after some calculations that equating the $\DQAnalysis{.}{}$ of the two sides implies $p = p_1 = p_2 = 1$ or $p = p_1 = p_2 = 0$.
  When $(o_1 \ProbChoice{[p_1]}{} \bot) \MultiFTF (o_2 \ProbChoice{[p_2]}{} \bot)$ is $o_1 \MultiFTF o_2$, in which $o_1$ and $o_2$ are being properised, that is either when $o_1 = o_2 = \top$ or $o_1 = o_2 = \bot$.
\end{remark}

Hereafter, we will write $o_1 \ProbChoice{[.]}{} o_2$ to mean $o_1 \ProbChoice{[p]}{} o_2$ for some unimportant $p$.

The desirable inequalities in Theorem~\ref{Thrm:Dist.Ineq.Proper} are all of the form $o_l \neq o_r$, with the outcome variables in $o_l$ and $o_r$ being $o_1$, $o_2$, and $o_3$.
In order to show $o_l \neq o_r$, we proceed by properisation of $o_1$, $o_2$, and $o_3$ in $o_l$ and $o_r$.

To that end, we fix a basic assignment $\Delta$, such that $\DQAnalysis{o_k}{\Delta} = \delta_k$ and $\Im(\delta_k) = i_k$ for $k \in \{1, 2, 3\}$.
Then, we take $p_k = 1 - i_k$ for $k \in \{1, 2, 3\}$, $(o'_k, \Delta') = (o_k, \Delta)\Properise^{o_1, o_2, o_3}$ for $k \in \{l, r\}$.
We show that $\DQAnalysis{o'_l}{\Delta'} \neq \DQAnalysis{o'_r}{\Delta'}$ to conclude that
$\DQAnalysis{o_l}{\Delta} \neq \DQAnalysis{o_r}{\Delta}$ by Theorem~\ref{Thrm:Properise} and the result follows.

\begin{theorem} \label{Thrm:Dist.Ineq.Proper}
  For every $o_1, o_2, o_3 \in \Outcomes$,
  $$
  \begin{array}{l@{\hspace{4em}}l}
    (o_1 \SeqDelta o_2) \ProbChoice{}{} o_3 \neq (o_1 \ProbChoice{}{} o_3) \SeqDelta (o_2 \ProbChoice{}{} o_3) &  o_1 \ProbChoice{}{} (o_2 \SeqDelta o_3) \neq (o_1 \ProbChoice{}{} o_2) \SeqDelta (o_1 \ProbChoice{}{} o_3)\\
    (o_1 \MultiATF o_2) \ProbChoice{}{} o_3 \neq (o_1 \ProbChoice{}{} o_3) \MultiATF (o_2 \ProbChoice{}{} o_3) & o_1 \ProbChoice{}{} (o_2 \MultiATF o_3) \neq (o_1 \ProbChoice{}{} o_2) \MultiATF (o_1 \ProbChoice{}{} o_3)\\
    (o_1 \MultiATF o_2) \SeqDelta o_3 \neq (o_1 \SeqDelta o_3) \MultiATF (o_2 \SeqDelta o_3) & o_1 \SeqDelta (o_2 \MultiATF o_3) \neq (o_1 \SeqDelta o_2) \MultiATF (o_1 \SeqDelta o_3)\text{.}
  \end{array}
  $$
\end{theorem} 
\begin{proof}
  We only prove
  \begin{equation} \label{Eqtn:Dist.Ineq.Proper.1}
    (o_1 \SeqDelta o_2) \ProbChoice{[p]}{} o_3 \neq (o_1 \ProbChoice{[p]}{} o_3) \SeqDelta (o_2 \ProbChoice{[p]}{} o_3)
  \end{equation}
  for a given $p$ here.
  The rest can be proved similarly using Lemma~\ref{Lemm:Fail.Acc}.
  
  Fix a basic assignment $\Delta$, such that $\Im(\DQAnalysis{o_k}{\Delta}) = i_k$ for $k \in \{1, 2, 3\}$.
  Take $p_k = 1 - i_k$ for $k \in \{1, 2, 3\}$.
  Pick $o'_l$ and $o'_r$ such that $(o'_l, \Delta') = ((o_1 \SeqDelta o_2) \ProbChoice{[p]}{} o_3, \Delta)\Properise^{o_1, o_2, o_3}$ and $(o'_r, \Delta') = ((o_1 \ProbChoice{[p]}{} o_3) \SeqDelta (o_2 \ProbChoice{[p]}{} o_3), \Delta)\Properise^{o_1, o_2, o_3}$, for some basic assignment $\Delta'$.
  Our target inequality now becomes $o'_l \neq o'_r$, where
  \begin{itemize}
    \item $o'_l$ is $((o'_1 \ProbChoice{[p_1]}{} \bot) \SeqDelta (o'_2 \ProbChoice{[p_2]}{} \bot)) \ProbChoice{[p]}{} (o'_3 \ProbChoice{[p_3]}{} \bot)$, and
    \item $o'_r$ is $((o'_1 \ProbChoice{[p_1]}{} \bot) \ProbChoice{[p]}{} (o'_3 \ProbChoice{[p_3]}{} \bot)) \SeqDelta ((o'_2 \ProbChoice{[p_2]}{} \bot) \ProbChoice{[p]}{} (o'_3 \ProbChoice{[p_3]}{} \bot))$.
  \end{itemize}
  One can rewrite $o'_l$ using Lemma~\ref{Lemm:Fail.Acc} as $((o'_1 \SeqDelta o'_2) \ProbChoice{[.]}{} o'_3) \ProbChoice{[q]}{} \bot$, where $q = p_3 - pp_3 + pp_1p_2$.
  Likewise, $o'_r$ can be rewritten as $((o'_1 \ProbChoice{[.]}{} o'_3) \SeqDelta (o'_2 \ProbChoice{[.]}{} o'_3)) \ProbChoice{[r_1r_2]}{} \bot$, where $r_1 = p_3 - pp_3 + pp_1$ and $r_2 = p_3 - pp_3 + pp_2$.
  Should $o'_l \neq o'_r$ not hold, one gets $q = r_1r_2$.
  That is $p_3 - pp_3 + pp_1p_2 = (p_3 - pp_3 + pp_1)(p_3 - pp_3 + pp_2)$.
  But, that is not an equation that holds in general.
  Inequality~(\ref{Eqtn:Dist.Ineq.Proper.1}) follows by Theorem~\ref{Thrm:Properise}.
\end{proof}

\section{Related Work}

\dqsd{} has been used in practice by a small group of practitioners for a couple of decades now~\cite{TR-452.1,TR-452.2,Davi+Thom+Youn+Newt+Teig+Olde:2021,Teig+Davi+Olav+Skei+Torr:2022,Cout+Davi+Szam+Thom:2020}.
The first formalisation of \dqsd{} was, however, done quite recently by Haeri et al. \cite{Haer+Thom+Davi+Roy+Hamm+Chap:2022}.
We use that formalisation as a foundation.

Teigen et al \cite{Teig+Davi+Olav+Skei+Torr:2022} use \dq{} to develop a novel model of WiFi performance that produces complete latency distributions. 
The model is validated by comparison with previous modeling work and real-world measurements.
It would be very interesting to apply \dqsd{} to an outcome description of the protocol to see if this can replicate the same results.

Elsewhere, Gajda \cite{Gajd:2022} attempts to model latency distributions but allows operations that do not preserve total probability, hence, leading to incorrect conclusions about failure probabilities.

Business Process Modelling and Notation (BPMN) \cite{Sher:2012} is a diagram scheme which is closely related to Outcome Diagrams 
(although with some details that are not considered relevant to \dqsd{}).
BPMN supports all \dqsd{} operators except probabilistic choice.
The closest operator is their ``xor'' gateway, which is essentially $\ProbChoice{[0.5]}{}$.
It is less expressive to the extent that it makes it impossible to consider systems such as the example in Section~\ref{Sect:Cache}.
Of the attempts for formalising BPMN, those of Wong and Gibbons \cite{Wong+Gibb:2011-1,Wong+Gibb:2011-2} are the most related to our work.
Wong and Gibbons use the CSP process algebra for that purpose and further develop it to enable the specification of timing constraints on concurrent systems.
Their developments allow mechanical verification of behavioural properties of BPMN diagrams using the FDR2 \cite{FDR2:2012} refinement checker.
Whilst Wong and Gibbons prove many interesting properties of their BPMN instances, they do not consider algebraic equivalences or algebraic structures for BPMN as we do in this work for \dqsd{}.
A less related BPMN formalisation work is that of El Hichami et al. \cite{Hich+Naou+Achha+Berr+Moha:2015}, which provides a denotational semantics based on the Max+ algebra as an execution model for BPMN.
They list a handful of algebraic equivalences in Max+ only axiomatically.
Nevertheless, El Hichami et al. make no attempt to study the equivalence of BPMN diagrams based on their Max+ semantics.

When it comes to timeliness analysis, an important advantage of outcome diagrams over BPMNs is Definition~\ref{Defn:DQ.Analysis.Composition}, which formally defines the timeliness analysis of outcome diagrams.
Definition~\ref{Defn:DQ.Analysis.Composition} is fundamental to the applicability of the model theory we employ in this paper (Section~\ref{Sect:Model.Theory}).
We are not aware of any formally defined recipe for timeliness analysis of BPMNs.
The two closest attempts that we could find are the following two:
Friedenstab et al. \cite{Frie+Jani+Matz+Mull:2012} borrow constructs from Business Activity Monitoring \cite{Cost+Moll:2008} to augment BPMN with a graphical notation for describing certain timeliness matters.
Likewise, Morales \cite{Mora:2014} informally describes how to transform BPMN diagrams to timed automata networks, suggesting qualitative analysis of timeliness. 

Performance Evaluation Process Algebra (PEPA) \cite{PEPA} is an algebraic language for performance modelling of systems.
PEPA is successful and well-published with a rich family of formalisations with various interesting theoretical properties.
However, PEPA suffers from several shortcomings that make it difficult to apply to real-world software systems.
For example, PEPA does not model open or partially-specified systems; every detail of the system needs to be determined in advance.
Since PEPA does not allow goals and objectives to be specified, it offers no assistance when comparing the predicted performance with the requirements.
PEPA also suffers from state explosion, rapidly making it impractical, although more recent PEPA technology employs continuous approximations of the states, which contain some of the state explosion.
This is similar to the use of IRVs in \dqsd{} but rather \emph{ad hoc} compared with the systematic use of \dq{}s in \dqsd{}.
Less conservative alternatives to PEPA like SCEL \cite{Nico+Late+Lafu+Lore+Marg+Mass+Mori+Pugl+Tiez+Vand:2015} allow open systems but suffer from even more state explosion.
CARMA \cite{Bort+Nico+Galp+Gilm+Hill+Late+Lore+Mass:2015} addresses a lot of the problems with PEPA, using a fluid approximation to manage the state explosion.

PerformERL \cite{Cazz+Cesa+Tran:2022} is an Erlang toolset, which focuses on monitoring the relationship 
between load repeatability and internal resource allocation.
The authors advertise their toolset as an assistant for making early stage performance decisions,
but it is unclear how it does this.
Uunlike \dqsd{}, monitoring (like testing) requires implementation of the system specification up to a certain level.
The closer the implementation is to the full specification, the more reliable the monitoring will become, but the analysis is then no longer early-stage.
Less accurate monitoring, on the other hand, is not reliable for decision making.
The closest PerformERL gets to the work described in this paper is its lightweight theoretical work out of the monitoring overhead it imposes to the system under development.

Finally, Failure Modes Effects Analysis \cite{FMEA} (FMEA) considers how failures propagate through a system but, unlike \dqsd{}, does not model delays.
We are not aware of any formalisation of FMEA that can serve algebraic developments like those on failure in this paper.

\section{Conclusion and Future Work} \label{Sect:Future.Work}

This paper lays down model-theoretic foundations for timeliness analysis \`a la \dqsd{}.
It establishes time as a resource that is consumed by outcomes.
In doing so, it enables timeliness analysis \emph{via} the study of quality attenuation, simultaneously capturing both delay and failure.
With our focus being exclusively on timeliness, we discuss the algebraic structures that the \dqsd{} operators form with outcome expressions (Theorems~\ref{Thrm:ProbChoice.Magma}--\ref{Thrm:FTF.Comm.Monoid}).
We refute the formation of richer algebraic structures by the \dqsd{} operators and outcome expressions (Remarks~\ref{Rmrk:ATF.Not.Group}, \ref{Rmrk:FTF.Not.Group}, and \ref{Rmrk:FTF.ATF.No.Semiring}).
We consider the $15$ distributivity results about the \dqsd{} operators.
We prove $3$ (Theorem~\ref{Thrm:Dist.Prob}) and disprove $8$ (Theorem~\ref{Thrm:Dist.Ineq.Proper}) using the newly formalised technique developed in this paper called properisation (Theorem~\ref{Thrm:Properise}) and $4$ using counterexamples (Theorem~\ref{Thrm:Dist.Fail.Counter}).
We also provide guidelines for studying the existence of potential distributivity (Section~\ref{Sect:Potential.Distributivity}).
Finally, we establish $14$ important equivalences that have already been used in the practice of \dqsd{} over the past few decades (Lemmas~\ref{Lemm:Prob.Choice.Assoc.Left.Coef}--\ref{Lemm:Prob.Choice.Assoc.Right.Coef} and Theorem~\ref{Thrm:Fig.Equivalences}).

Our immediate future work is to study the algebraic properties of other resources \`a la \dqsd{}, with the eventual goal of
providing
an algebraic categorisation of resources.
A sound theoretical foundation is essential for the construction of robust tool support, which is, in turn, a prerequisite for wider application of the \dqsd{} paradigm.
Currently, there is a numerically-based tool prototype.
However, to deal effectively with large complex systems, this needs to be made more symbolic.
The aim is for the expressions to be simplified before calculation, and to be able to represent performance unknowns. 
Algebraic structures are essential for correctly manipulating and simplifying expressions.
This work informs both ongoing practical work and tool development.
Conversely, consideration of specific aspects of system design and operation will inform the most productive directions for the theoretical developments. 

To conclude, this paper has introduced a number of important algebraic properties for \dqsd{} outcome expressions.
These properties have a highly practical application in the analysis of timeliness and resource consumption.
For the first time, we have shown distributivity of the \dqsd{} operators over probabilistic choice, and placed a set of `folklore' equivalences (Theorem~\ref{Thrm:Fig.Equivalences}) that are in common usage for \dqsd{} on a sound footing.
These equivalences are essential for rapid recognition of infeasibility and for sound manipulation of outcome expressions to reduce computational complexity.

\section*{Acknowledgements}

This research is funded by IOG, Singapore as a part of an ongoing project for incorporating performance as a first-class factor of the software development life cycle.
When the routine proof technique did not work for distributivity, Andre Knispel (of IOG) suggested that we could utilise easier properties to obtain the disproofs using contrapositive reasoning.
We would like to thank him for that suggestion.

\bibliographystyle{eptcs}
\bibliography{Delta-Q-bibliography,AdditionalEntries}

\begin{thebibliography}{10}
\providecommand{\bibitemdeclare}[2]{}
\providecommand{\surnamestart}{}
\providecommand{\surnameend}{}
\providecommand{\urlprefix}{Available at }
\providecommand{\url}[1]{\texttt{#1}}
\providecommand{\href}[2]{\texttt{#2}}
\providecommand{\urlalt}[2]{\href{#1}{#2}}
\providecommand{\doi}[1]{doi:\urlalt{https://doi.org/#1}{#1}}
\providecommand{\eprint}[1]{arXiv:\urlalt{https://arxiv.org/abs/#1}{#1}}
\providecommand{\bibinfo}[2]{#2}

\bibitemdeclare{techreport}{FMEA}
\bibitem{FMEA}
 (\bibinfo{year}{1980}): \emph{\bibinfo{title}{MIL-STD-1629A -- Procedures for
  Performing a Failure Mode Effect and Criticality Analysis}}.
\newblock \bibinfo{type}{Technical Report}, \bibinfo{institution}{United States
  Department of Defense}.

\bibitemdeclare{inproceedings}{Bort+Nico+Galp+Gilm+Hill+Late+Lore+Mass:2015}
\bibitem{Bort+Nico+Galp+Gilm+Hill+Late+Lore+Mass:2015}
\bibinfo{author}{L.~\surnamestart Bortolussi\surnameend},
  \bibinfo{author}{R.~\surnamestart De~Nicola\surnameend},
  \bibinfo{author}{V.~\surnamestart Galpin\surnameend},
  \bibinfo{author}{S.~\surnamestart Gilmore\surnameend},
  \bibinfo{author}{J.~\surnamestart Hillston\surnameend},
  \bibinfo{author}{D.~\surnamestart Latella\surnameend},
  \bibinfo{author}{M.~\surnamestart Loreti\surnameend} \&
  \bibinfo{author}{M.~\surnamestart Massink\surnameend} (\bibinfo{year}{2015}):
  \emph{\bibinfo{title}{{CARMA:} Collective Adaptive Resource-sharing Markovian
  Agents}}.
\newblock In \bibinfo{editor}{N.~\surnamestart Bertrand\surnameend} \&
  \bibinfo{editor}{M.~\surnamestart Tribastone\surnameend}, editors: {\slshape
  \bibinfo{booktitle}{Proc. $13^\mathit{th}$ W. Quant. Aspects of Prog. Lang.
  and Sys.}}, {\slshape \bibinfo{series}{{EPTCS}}} \bibinfo{volume}{194}, pp.
  \bibinfo{pages}{16--31}, \doi{10.4204/EPTCS.194.2}.

\bibitemdeclare{article}{Cazz+Cesa+Tran:2022}
\bibitem{Cazz+Cesa+Tran:2022}
\bibinfo{author}{W.~\surnamestart Cazzola\surnameend},
  \bibinfo{author}{F.~\surnamestart Cesarini\surnameend} \&
  \bibinfo{author}{L.~\surnamestart Tansini\surnameend} (\bibinfo{year}{2022}):
  \emph{\bibinfo{title}{{PerformERL: A Performance Testing Framework for
  Erlang}}}.
\newblock {\slshape \bibinfo{journal}{Distributed Comp.}}
  \bibinfo{volume}{35}(\bibinfo{number}{5}), pp. \bibinfo{pages}{439--454},
  \doi{10.1007/s00446-022-00429-7}.

\bibitemdeclare{inproceedings}{Cost+Moll:2008}
\bibitem{Cost+Moll:2008}
\bibinfo{author}{C.~\surnamestart Costello\surnameend} \&
  \bibinfo{author}{O.~\surnamestart Molloy\surnameend} (\bibinfo{year}{2008}):
  \emph{\bibinfo{title}{{Towards a Semantic Framework for Business Activity
  Monitoring and Management}}}.
\newblock In: {\slshape \bibinfo{booktitle}{{AAAI Spring Symposium: AI meets
  business rules and process management}}}, pp. \bibinfo{pages}{17--27}.

\bibitemdeclare{techreport}{Cout+Davi+Szam+Thom:2020}
\bibitem{Cout+Davi+Szam+Thom:2020}
\bibinfo{author}{D.~\surnamestart Coutts\surnameend},
  \bibinfo{author}{N.~\surnamestart Davies\surnameend},
  \bibinfo{author}{M.~\surnamestart Szamotulski\surnameend} \&
  \bibinfo{author}{P.~\surnamestart Thompson\surnameend}
  (\bibinfo{year}{2020}): \emph{\bibinfo{title}{{Introduction to the Design of
  the Data Diffusion and Networking for Cardano Shelley}}}.
\newblock \bibinfo{type}{Technical Report}, \bibinfo{institution}{IOHK}.
\newblock
  \urlprefix\url{https://hydra.iohk.io/build/20405228/download/1/network-design.pdf}.

\bibitemdeclare{incollection}{Davi+Thom+Youn+Newt+Teig+Olde:2021}
\bibitem{Davi+Thom+Youn+Newt+Teig+Olde:2021}
\bibinfo{author}{N.~\surnamestart Davies\surnameend},
  \bibinfo{author}{P.~\surnamestart Thompson\surnameend},
  \bibinfo{author}{G.~\surnamestart Young\surnameend},
  \bibinfo{author}{J.~\surnamestart Newton\surnameend},
  \bibinfo{author}{B.~\surnamestart Teigen\surnameend} \&
  \bibinfo{author}{M.~\surnamestart Olden\surnameend} (\bibinfo{year}{2021}):
  \emph{\bibinfo{title}{Measuring Network Impact on Application Outcomes Using
  Quality Attenuation}}.
\newblock In: {\slshape \bibinfo{booktitle}{{Measuring Network Quality for
  End-Users}}}, \bibinfo{publisher}{Internet Architecture Board}, pp.
  \bibinfo{pages}{43--52}.
\newblock
  \urlprefix\url{https://www.iab.org/wp-content/IAB-uploads/2021/09/PNSol-et-al-Submission-to-Measuring-Network-Quality-for-End-Users-1.pdf}.

\bibitemdeclare{inbook}{Nico+Late+Lafu+Lore+Marg+Mass+Mori+Pugl+Tiez+Vand:2015}
\bibitem{Nico+Late+Lafu+Lore+Marg+Mass+Mori+Pugl+Tiez+Vand:2015}
\bibinfo{author}{R.~\surnamestart De~Nicola\surnameend},
  \bibinfo{author}{D.~\surnamestart Latella\surnameend}, \bibinfo{author}{A.~L.
  \surnamestart Lafuente\surnameend}, \bibinfo{author}{M.~\surnamestart
  Loreti\surnameend}, \bibinfo{author}{A.~\surnamestart Margheri\surnameend},
  \bibinfo{author}{M.~\surnamestart Massink\surnameend},
  \bibinfo{author}{A.~\surnamestart Morichetta\surnameend},
  \bibinfo{author}{R.~\surnamestart Pugliese\surnameend},
  \bibinfo{author}{F.~\surnamestart Tiezzi\surnameend} \&
  \bibinfo{author}{A.~\surnamestart Vandin\surnameend} (\bibinfo{year}{2015}):
  \emph{\bibinfo{title}{{The SCEL Language: Design, Implementation,
  Verification}}}, pp. \bibinfo{pages}{3--71}.
\newblock \bibinfo{publisher}{Springer}, \doi{10.1007/978-3-319-16310-9_1}.

\bibitemdeclare{article}{Hich+Naou+Achha+Berr+Moha:2015}
\bibitem{Hich+Naou+Achha+Berr+Moha:2015}
\bibinfo{author}{O.~\surnamestart El~Hichami\surnameend},
  \bibinfo{author}{M.~\surnamestart Naoum\surnameend},
  \bibinfo{author}{M.~\surnamestart Al~Achhab\surnameend},
  \bibinfo{author}{I.~\surnamestart Berrada\surnameend} \&
  \bibinfo{author}{B.~E. \surnamestart El~Mohajir\surnameend}
  (\bibinfo{year}{2015}): \emph{\bibinfo{title}{{An Algebraic Method for
  Analysing Control Flow of BPMN Models}}}.
\newblock {\slshape \bibinfo{journal}{iJES}}
  \bibinfo{volume}{3}(\bibinfo{number}{3}), pp. \bibinfo{pages}{20--–26},
  \doi{10.3991/ijes.v3i3.4862}.
\newblock
  \urlprefix\url{https://online-journals.org/index.php/i-jes/article/view/4862}.

\bibitemdeclare{inproceedings}{Frie+Jani+Matz+Mull:2012}
\bibitem{Frie+Jani+Matz+Mull:2012}
\bibinfo{author}{J.-P. \surnamestart Friedenstab\surnameend},
  \bibinfo{author}{C.~\surnamestart Janiesch\surnameend},
  \bibinfo{author}{M.~\surnamestart Matzner\surnameend} \&
  \bibinfo{author}{O.~\surnamestart Muller\surnameend} (\bibinfo{year}{2012}):
  \emph{\bibinfo{title}{{Extending BPMN for Business Activity Monitoring}}}.
\newblock In: {\slshape \bibinfo{booktitle}{$45^\mathit{th}$ HICSS}}, pp.
  \bibinfo{pages}{4158--4167}, \doi{10.1109/HICSS.2012.276}.

\bibitemdeclare{article}{Gajd:2022}
\bibitem{Gajd:2022}
\bibinfo{author}{M.~J. \surnamestart Gajda\surnameend} (\bibinfo{year}{2020}):
  \emph{\bibinfo{title}{{Curious Properties of Latency Distributions}}}.
\newblock {\slshape \bibinfo{journal}{CoRR}} \bibinfo{volume}{abs/2011.05219},
  \doi{10.1007/978-3-031-10461-9_10}.
\newblock \urlprefix\url{https://arxiv.org/abs/2011.05219}.

\bibitemdeclare{article}{Haer+Thom+Davi+Roy+Hamm+Chap:2022}
\bibitem{Haer+Thom+Davi+Roy+Hamm+Chap:2022}
\bibinfo{author}{S.~H. \surnamestart Haeri\surnameend},
  \bibinfo{author}{P.~\surnamestart Thompson\surnameend},
  \bibinfo{author}{N.~\surnamestart Davies\surnameend},
  \bibinfo{author}{P.~\surnamestart Van~Roy\surnameend},
  \bibinfo{author}{K.~\surnamestart Hammond\surnameend} \&
  \bibinfo{author}{J.~\surnamestart Chapman\surnameend} (\bibinfo{year}{2022}):
  \emph{\bibinfo{title}{{Mind Your Outcomes: The $\Delta$QSD Paradigm for
  Quality-Centric Systems Development and Its Application to a Blockchain Case
  Study}}}.
\newblock {\slshape \bibinfo{journal}{Computers}}
  \bibinfo{volume}{11}(\bibinfo{number}{3}), p.~\bibinfo{pages}{45},
  \doi{10.3390/computers11030045}.
\newblock \urlprefix\url{https://www.mdpi.com/2073-431X/11/3/45}.

\bibitemdeclare{techreport}{Haer+Thom+Roy+Have+Davi+Bara+Chap:2023}
\bibitem{Haer+Thom+Roy+Have+Davi+Bara+Chap:2023}
\bibinfo{author}{S.~H. \surnamestart Haeri\surnameend}, \bibinfo{author}{P.~W.
  \surnamestart Thompson\surnameend}, \bibinfo{author}{P.~\surnamestart
  Van~Roy\surnameend}, \bibinfo{author}{M.~\surnamestart Haveraaen\surnameend},
  \bibinfo{author}{N.~J. \surnamestart Davies\surnameend},
  \bibinfo{author}{M.~\surnamestart Barash\surnameend} \&
  \bibinfo{author}{J.~\surnamestart Chapman\surnameend} (\bibinfo{year}{2023}):
  \emph{\bibinfo{title}{On the Algebraic Properties of Timeliness}}.
\newblock \bibinfo{type}{Technical Report}, \bibinfo{institution}{IOG}.
\newblock
  \urlprefix\url{http://www.pnsol.com/public/Algebraic-Timeliness-TR.pdf}.

\bibitemdeclare{book}{PEPA}
\bibitem{PEPA}
\bibinfo{author}{J.~\surnamestart Hillston\surnameend} (\bibinfo{year}{1996}):
  \emph{\bibinfo{title}{A Compositional Approach to Performance Modelling}}.
\newblock \bibinfo{publisher}{Cambridge University Press},
  \doi{10.1017/CBO9780511569951}.

\bibitemdeclare{inproceedings}{Teig+Davi+Olav+Skei+Torr:2022}
\bibitem{Teig+Davi+Olav+Skei+Torr:2022}
\bibinfo{author}{B.~\surnamestart Ivar~Teigen\surnameend},
  \bibinfo{author}{N.~\surnamestart Davies\surnameend},
  \bibinfo{author}{K.~\surnamestart Olav~Ellefsen\surnameend},
  \bibinfo{author}{T.~\surnamestart Skeie\surnameend} \&
  \bibinfo{author}{J.~\surnamestart Torresen\surnameend}
  (\bibinfo{year}{2022}): \emph{\bibinfo{title}{{Quantifying the Quality
  Attenuation of WiFi}}}.
\newblock In \bibinfo{editor}{S.~\surnamestart Oteafy\surnameend},
  \bibinfo{editor}{E.~\surnamestart Bulut\surnameend} \&
  \bibinfo{editor}{F.~\surnamestart Tschorsch\surnameend}, editors: {\slshape
  \bibinfo{booktitle}{IEEE $47^\mathit{th}$ LCN}}, \bibinfo{publisher}{{IEEE}},
  pp. \bibinfo{pages}{189--197}, \doi{10.1109/LCN53696.2022.9843690}.

\bibitemdeclare{misc}{FDR2:2012}
\bibitem{FDR2:2012}
\bibinfo{author}{Formal Systems~(Europe) \surnamestart Ltd\surnameend}
  (\bibinfo{year}{2012}): \emph{\bibinfo{title}{Failures-Divergence Refinement:
  FDR2 User Manual}}.
\newblock
  \urlprefix\url{https://www.cs.ox.ac.uk/projects/concurrency-tools/download/fdr2manual-2.94.pdf}.

\bibitemdeclare{inproceedings}{Mora:2014}
\bibitem{Mora:2014}
\bibinfo{author}{L.~E.~M. \surnamestart Morales\surnameend}
  (\bibinfo{year}{2014}): \emph{\bibinfo{title}{{Specifying BPMN Diagrams with
  Timed Automata: Proposal of Some Mapping Rules}}}.
\newblock In: {\slshape \bibinfo{booktitle}{$9^\mathit{th}$ CISTI}}, pp.
  \bibinfo{pages}{1--6}, \doi{10.1109/CISTI.2014.6876897}.

\bibitemdeclare{misc}{pnsol}
\bibitem{pnsol}
\bibinfo{author}{\surnamestart {Predictable Network Solutions Ltd
  (PNSol)}\surnameend} (\bibinfo{year}{2022}):
  \urlprefix\url{http://www.pnsol.com}.

\bibitemdeclare{book}{Sher:2012}
\bibitem{Sher:2012}
\bibinfo{author}{K.~J. \surnamestart Sherry\surnameend} (\bibinfo{year}{2012}):
  \emph{\bibinfo{title}{{Business Process Modelling with BPMN: Modelling and
  Designing Business Processes Course Book using The Business Process Model and
  Notation Specification Version 2.0}}}.
\newblock \bibinfo{publisher}{CreateSpace Independent Publishing Platform}.

\bibitemdeclare{techreport}{TR-452.2}
\bibitem{TR-452.2}
\bibinfo{author}{P.~\surnamestart Thompson\surnameend} (\bibinfo{year}{2022}):
  \emph{\bibinfo{title}{TR-452.2 Quality Attenuation Measurements using Active
  Test Protocols}}.
\newblock \bibinfo{type}{Technical Report}, \bibinfo{institution}{The Broadband
  Forum}.

\bibitemdeclare{techreport}{TR-452.1}
\bibitem{TR-452.1}
\bibinfo{author}{P.~\surnamestart Thompson\surnameend} \&
  \bibinfo{author}{R.~\surnamestart Hernadaz\surnameend}
  (\bibinfo{year}{2020}): \emph{\bibinfo{title}{Quality Attenuation Measurement
  Architecture and Requirements}}.
\newblock \bibinfo{type}{Technical Report} \bibinfo{number}{TR-452.1},
  \bibinfo{institution}{Broadband {F}orum}.
\newblock
  \urlprefix\url{https://www.broadband-forum.org/download/TR-452.1.pdf}.

\bibitemdeclare{book}{Triv:1982}
\bibitem{Triv:1982}
\bibinfo{author}{K.~S. \surnamestart Trivedi\surnameend}
  (\bibinfo{year}{2002}): \emph{\bibinfo{title}{{Probability and Statistics
  with Reliability, Queuing, and Computer Science Applications}}},
  \bibinfo{edition}{2} edition.
\newblock \bibinfo{publisher}{Wiley}, \bibinfo{address}{New York, NY, USA}.

\bibitemdeclare{techreport}{HiPEAC:2023}
\bibitem{HiPEAC:2023}
\bibinfo{author}{P.~\surnamestart {Van Roy}\surnameend},
  \bibinfo{author}{N.~\surnamestart Davies\surnameend},
  \bibinfo{author}{P.~\surnamestart Thompson\surnameend} \&
  \bibinfo{author}{S.~H. \surnamestart Haeri\surnameend}
  (\bibinfo{year}{2023}): \emph{\bibinfo{title}{{$\Delta$QSD}: Designing
  Systems with Predictable Latency at High Load}}.
\newblock \bibinfo{type}{Tutorial}, \bibinfo{institution}{HiPEAC 2023 (Conf.
  High Perf. Emb. Arch. \& Compil.)}.
\newblock \urlprefix\url{shorturl.at/dmKSW}.

\bibitemdeclare{article}{Wong+Gibb:2011-1}
\bibitem{Wong+Gibb:2011-1}
\bibinfo{author}{P.~Y.~H. \surnamestart Wong\surnameend} \&
  \bibinfo{author}{J.~\surnamestart Gibbons\surnameend} (\bibinfo{year}{2011}):
  \emph{\bibinfo{title}{{Formalisations and Applications of BPMN}}}.
\newblock {\slshape \bibinfo{journal}{SCP}}
  \bibinfo{volume}{76}(\bibinfo{number}{8}), pp. \bibinfo{pages}{633--650},
  \doi{10.1016/j.scico.2009.09.010}.
\newblock
  \urlprefix\url{https://www.sciencedirect.com/science/article/pii/S0167642309001282}.

\bibitemdeclare{article}{Wong+Gibb:2011-2}
\bibitem{Wong+Gibb:2011-2}
\bibinfo{author}{P.~Y.~H. \surnamestart Wong\surnameend} \&
  \bibinfo{author}{J.~\surnamestart Gibbons\surnameend} (\bibinfo{year}{2011}):
  \emph{\bibinfo{title}{{Property Specifications for Workflow Modelling}}}.
\newblock {\slshape \bibinfo{journal}{SCP}}
  \bibinfo{volume}{76}(\bibinfo{number}{10}), pp. \bibinfo{pages}{942--967},
  \doi{10.1016/j.scico.2010.09.007}.
\newblock
  \urlprefix\url{https://www.sciencedirect.com/science/article/pii/S0167642310001735}.

\end{thebibliography}
\end{document}